\documentclass[a4paper,11pt]{article}

\usepackage[english]{babel} 
\usepackage[T1]{fontenc}
\usepackage{amsfonts}
\usepackage{amsmath, amssymb}
\usepackage{amsthm}
\usepackage{enumerate}
\usepackage{cases}
\usepackage{pdfsync}
\usepackage{graphicx}
\usepackage{subfigure}
\usepackage{mathrsfs}
\usepackage{dashundergaps}
\usepackage{fancyhdr}
\usepackage{array}
\usepackage{fullpage}
\usepackage[dvipsnames]{xcolor}
\usepackage{multirow}
\usepackage{hhline}
\usepackage{rotating}
\usepackage[authoryear]{natbib}
\usepackage[colorlinks=true,urlcolor=black,linkcolor=NavyBlue,citecolor=NavyBlue]{hyperref}
 
\newcommand{\R}{\mathbb R}
\newcommand{\quant}[1]{\widetilde #1^N}
\newcommand{\tildeq}{\widetilde q^N_\alpha(x)}

\newcommand{\hatq}{\widehat q^{N,n}_{\alpha}(x)}

\newcommand{\barq}{\bar {q}^{N,n}_{\alpha,B}(x)}

\newcommand{\suppX}{S_X}

\newtheorem{theorem}{Theorem}[section]
\newtheorem{lemma}[theorem]{Lemma}

\newtheorem{corollary}[theorem]{Corollary}

\theoremstyle{definition}

\makeatletter

\@addtoreset{equation}{section}
\makeatother

\makeatletter
\newcommand\ackname{Acknowledgements}
\if@titlepage
  \newenvironment{acknowledgements}{%
      \titlepage
      \null\vfil
      \@beginparpenalty\@lowpenalty
      \begin{center}%
        \bfseries \ackname
        \@endparpenalty\@M
      \end{center}}%
     {\par\vfil\null\endtitlepage}
\else
  
\fi
\makeatother

\begin{document}
\title{\Large
	{\textbf{\textsc{Conditional Quantile Estimation \\
	through Optimal Quantization}}}}
\author{\sc{Isabelle Charlier$^{(1,2,3)}\footnote{Research is supported by a Bourse F.R.I.A. of the Fonds National de la Recherche Scientifique, Communaut\'e fran\c{c}aise de Belgique.}$}\and\sc{ Davy Paindaveine$^{(1,2)}$}\footnote{Research is supported by an \mbox{A.R.C.} contract from the Communaut\'e Fran\c{c}aise de Belgique and by the IAP research network grant P7/06 of the Belgian government (Belgian Science Policy).} \and \sc{J\'er\^ome Saracco$^{(3)}$}}
\maketitle
\begin{center}
{\it
$^{(1)}$ Universit\'e Libre de Bruxelles, D\'epartement de Math\'ematique, Boulevard du Triomphe, Campus Plaine, CP210, B-1050, Bruxelles, Belgique.

\textnormal{\href{mailto:ischarli@ulb.ac.be}{\nolinkurl{ischarli@ulb.ac.be}}}, \textnormal{\href{mailto:dpaindav@ulb.ac.be}{\nolinkurl{dpaindav@ulb.ac.be}}}
 
$^{(2)}$ ECARES, 50 Avenue F.D. Roosevelt, CP114/04, B-1050, Bruxelles, Belgique.

$^{(3)}$ Universit\'e de Bordeaux, Institut de Math\'ematiques de Bordeaux, UMR CNRS 5251 et INRIA Bordeaux Sud-Ouest, \'equipe CQFD, 351 Cours de la Lib\'eration, 33405 Talence.

\textnormal{\href{mailto:Jerome.Saracco@math.u-bordeaux1.fr}{\nolinkurl{Jerome.Saracco@math.u-bordeaux1.fr}}}
}
\end{center}

\begin{abstract}
In this paper, we use quantization to construct a nonparametric estimator of conditional quantiles of a scalar response~$Y$ given a $d$-dimensional vector of covariates~$X$. First we focus on the population level and show how optimal quantization of~$X$, which consists in discretizing~$X$ by projecting it on an appropriate grid of~$N$ points, allows to approximate conditional quantiles of~$Y$ given~$X$. We show that this is approximation is arbitrarily good as~$N$ goes to infinity and provide a rate of convergence for the approximation error. Then we turn to the sample case and define an estimator of conditional quantiles based on quantization ideas. We prove that this estimator is consistent for its fixed-$N$ population counterpart. The results are illustrated on a numerical example. Dominance of our estimators over local constant/linear ones and nearest neighbor ones is demonstrated through extensive simulations in the companion paper \citet{Chaetal2014b}.
\end{abstract}


\section{Introduction}

In numerous applications, one considers regression modelling to assess the impact of a $d$-dimensional vector of covariates~$X$ on a scalar response variable~$Y$. It is then classical to consider the conditional mean and variance functions
\begin{equation}
\label{condmeanvar}
x\mapsto {\rm E}[Y|X=x]
\quad 
\textrm{ and }
\quad 
x\mapsto {\rm Var}[Y|X=x],
\end{equation}
respectively. A much more thorough picture, however, is obtained by considering, for various $\alpha\in(0,1)$, the conditional quantile functions 
\begin{equation}
\label{condquant}
x\mapsto q_\alpha(x)= \inf\big\{y\in\R : F(y|x) \ge \alpha\big\}
,
\end{equation}
where $F(\,\cdot\,|x)$ denotes the conditional distribution of $Y$ given~$X=x$. These conditional quantile functions completely characterize the conditional distribution of~$Y$ given~$X$, whereas~(\ref{condmeanvar}), in contrast, only measures the impact of~$X$ on~$Y$'s location and scale, hence may completely miss to capture a possible impact of~$X$ on the shape of $Y$'s distribution, for instance. 

An important application of conditional quantiles is that they provide reference curves or surfaces (the graphs of~$x\mapsto q_\alpha(x)$ for various $\alpha$) and conditional prediction intervals (intervals of the form~$I_\alpha(x)=[q_{\alpha}(x),q_{1-\alpha}(x)]$, for fixed~$x$) that are widely used in many different areas. In medicine,  reference growth curves for children's height and weight as a function of age are considered. Reference curves are also of high interest in economics (e.g., to study discrimination effects and trends in income inequality),  in ecology (to observe how some covariates can affect limiting sustainable population size), and in lifetime analysis (to assess influence of risk factors on survival curves), among many others.

Quantile regression, that concerns the estimation of conditional quantile curves, was introduced in the seminal paper \citet{KoenkerBassett}, where the focus was on linear regression. Since then, there has been much research on quantile regression, in particular in the nonparametric regression framework. Kernel and nearest-neighbor estimators of conditional quantiles were investigated in \citet{BhattachGango}, while \citet{YuJones2} focused on local linear quantile regression and double-kernel approaches. Many other estimators were also considered; see, among others, \citet{Fan_et_al}, \citet{Gannoun_Saracco}, \citet{HeagertyPepe}, or \citet{Yu_al}.
In this work, we introduce a new nonparametric regression quantile method, based on \emph{optimal quantization}.

Optimal quantization is a tool that was first used by engineers in signal and information theory, where ``quantization'' refers to the discretization of a continuous signal using a finite number of points, called \emph{quantizers}. The aim being to achieve an efficient, parsimonious, transmission of the signal, the number and location of the quantizers have to be optimized. Quantization was later used in cluster analysis, pattern and speech recognition. More recently, it was considered in probability theory;
see, e.g., \citet{Zador} or \citet{Pages98}. In this context, the problem of optimal quantization consists in finding the best approximation of a continuous $d$-dimensional probability distribution~$P$ by a discrete probability distribution charging a fixed number~$N$ of points. In other words, the $d$-dimensional random vector~$X$ needs to be approximated by a random vector $\quant{X}$ that may assume at most $N$ values. 
Quantization was extensively investigated in numerical probability, finance, stochastic processes, and numerical integration~; see, e.g., \citet{Pages_al2004a}, \citet{Pages_al2004b}, or \citet{Bally_al}. 

Quantization, however, was seldom used in statistics. To the best of our knowledge, its applications in statistics are restricted to Sliced Inverse Regression (\citealp{Azais_al}) and clustering (\citealp{Fis2010} and \citealp{Fis2013}). As announced above, we use in this paper quantization in a nonparametric quantile regression framework. In this context, indeed, quantization naturally takes care of the localization-in-$x$ required in any nonparametric regression method. The resulting quantization-based estimators inherently are based on adaptive bandwidths, hence may dominate the local constant and local linear estimators~\citet{YuJones2} that typically involve a unique global bandwidth. Quantization-based estimators also provide a refinement over nearest-neighbor estimators (such as those from \citet{BhattachGango}) since, unlike the latter, the number of ``neighbors'' the former consider depends on the point~$x$ at which $q_\alpha(x)$ is to be estimated. This dominance over these two standard competitors, in terms of MSEs, is demonstrated through extensive simulations in the companion paper \citet{Chaetal2014b}. 

The outline of the paper, that mostly focuses on theoretical aspects, is as follows. Section~\ref{subsect_quantif} discusses quantization and provides some results on quantization, both of a theoretical and algorithmic nature. Section~\ref{sect_approx} describes how to approximate conditional quantiles through optimal quantization, which is achieved by replacing~$X$ in the definition of conditional quantiles by its ${\text L}_p$-optimal quantized version $\quant{X}$ (for some fixed~$N$). The convergence rate of this approximation to the true conditional quantiles is obtained. Section~\ref{sect_estim} defines the corresponding estimator and proves its consistency (for the fixed-$N$ approximated conditional quantiles). The results are illustrated on a numerical example, in which a smooth variant of the proposed estimator based on the bootstrap is also introduced. Section~\ref{sect_concl} provides some final comments. Eventually, the Appendix collects technical proofs.




%


\section{Optimal quantization}
\label{subsect_quantif}

In this section, we define the concept of ${\text L}_p$-norm optimal quantization and state the main results that will be used in the sequel (Section~\ref{secquant1}). Then we describe a stochastic algorithm that allows to perform optimal quantization (Section~\ref{subsubsect_algo}), and provide some convergence results for this algorithm (Section~\ref{subsubCLVQ}).

\subsection{Definition and main results}
\label{secquant1}

Let $X$ be a random $d$-vector defined on a probability space $(\Omega, \mathcal F, P)$, with distribution~$P_X$, and fix a real number~$p\geq 1$ such that~${\rm E}[|X|^p]<\infty$ (throughout, $|\cdot|$ denotes the Euclidean norm). Quantization replaces~$X$ with an appropriate random $d$-vector $\pi(X)$ that assumes at most~$N$ values. In optimal ${\text L}_p$-norm quantization, the vector~$\pi(X)$ minimizes the ${\text L}_p$-norm quantization error 
$$
\| \pi(X) - X \|_p
,
\quad
\textrm{ with }
\| Z \|_p
:=
\big({\rm E}\big[|Z|^p\big]\big)^{1/p}
.
$$ 
This optimization problem is equivalent to finding an $N$-grid of $\R^d$ ---  $\gamma^N$, say --- such that the projection $\widetilde X^{\gamma^N}=\text{Proj}_{\gamma^N}(X)$ of $X$ on the (Euclidean-)nearest point of the grid minimizes the \emph{quantization error} $\|\widetilde X^{\gamma^N}-X\|_p$. This definition leads to two natural questions: does such a minimum always exist? How does this minimum behave as~$N$ goes to infinity? 

Existence (but not unicity) of an optimal $N$-grid --- that is, a grid minimizing this quantization error --- has been obtained under the assumption that $P_X$ does not charge any hyperplane; see~\citet{Pages98}. Irrespective of the sequence of optimal grids considered, $\quant{X}$ converges to~$X$ in~$\text L_p$. This is a direct corollary of the following result, which is often referred to as Zador's theorem (\citealp{Zador}) and provides the rate of convergence of the quantization error; see, e.g., \citet{Graf_Lusch} for a proof.

 \begin{theorem}
   Assume that $\|X\|_{p+\delta}<\infty$ for some $\delta>0$. Let $P_X(du) = f(u)\lambda_d(du) + \nu(du)$ be the Lebesgue decomposition of $P_X$, where $\lambda_d$ is the Lebesgue measure on $\R^d$ and $\nu \bot \lambda_d$.   Then
   $$
   \lim_{N\to\infty} \left(N^{\frac{p}{d}}\min_{\gamma^N\in(\R^d)^N} \|\widetilde X^{\gamma^N}-X\|_p^p
\, \right) 
=
 J_{p,d}\left(\,\int_{\R^d} (f(u))^{\frac{d}{p+d}}\, du\right)^{1+\frac{p}{d}}
\!,
$$
with $J_{p,d}= \min_N \big( N^{p/d} \min_{\gamma^N\in(\R^d)^N} D_N^{p,U}(\gamma^N)\big)$, where $D^{p,U}_N(\gamma^N)$ denotes the $(p$th power of the) quantization error, obtained for the uniform distribution over~$[0,1]^d$, when considering the grid $\gamma^N\in(\R^d)^N$.
   \end{theorem}

In dimension~$d=1$, one has~$J_{p,d}=\frac{1}{2^p(p+1)}$. For~$d>1$, little is known about~$J_{p,d}$, but it can be shown that $J_{p,d} \sim \left(\frac{d}{2\pi e}\right)^{p/2}$ as $d\to\infty$; see \citet{Graf_Lusch}.

\begin{corollary}\label{rate_conv_qu}
	Assume that $\|X\|_{p+\delta}<\infty$ for some $\delta>0$. Then, for some~$C,D\in\R$ and~$N_0\in\mathbb{N}$, we have that
$$
\| \widetilde X^{\gamma^N} -X\|_p^p\le \frac{1}{N^{p/d}}
\big(C\|X\|^{p+\delta}_{p+\delta} + D\big),
$$
for all $N \ge N_0$.
\end{corollary}

Summing up, there exist ${\text L}_p$-optimal $N$-grids --- or \emph{optimal $N$-quantizers} --- that minimize the quantization error. A further natural question then is: how to obtain an optimal $N$-grid? We now discuss a stochastic algorithm that addresses this problem.

\subsection{The stochastic gradient algorithm}\label{subsubsect_algo}

Except in some very exceptional cases (such as the uniform over a compact interval of the real line), optimal $N$-grids have no closed form. That is, there exist  no results that describe the geometric structure of such grids. However, one can attempt to obtain (approximations of) optimal $N$-grids through a \emph{stochastic gradient algorithm} such as the following. 

Let $(\xi_t)_{t\in\mathbb{N}_0}$, $\mathbb{N}_0 = \{1,2,\ldots\}$, be a sequence of independent and identically $P_X$-distributed random vectors, and let $(\delta_t)_{t\in\mathbb{N}_0}$ be a deterministic sequence in~$(0,1)$ such that
    $$
    \sum_{t=1}^\infty\delta_t = +\infty \quad \textrm{and}\quad \sum_{t=1}^\infty\delta_t^2 < +\infty
    $$
(throughout the paper, we tacitly assume that the algorithm makes use of a sequence~$\delta_t$ that satisfies these conditions).  The algorithm starts from a deterministic $N$-tuple $X^0=x^0$ with $N$ pairwise distinct entries. This initial $N$-grid~$x^0$ is then updated
as follows. For every $t\in\mathbb{N}_0$, define recursively the grid $X^t$ as
\begin{equation}
\label{grad_sto_form}
X^t=X^{t-1}-\frac{\delta_t}{p}
\,
\nabla_x d^p_N(X^{t-1},\xi^t),
\end{equation}
where $\nabla_x d_N^p(x,\xi)$ stands for the gradient with respect to the $x$-argument of the so-called local quantization error
$
d^p_N(x,\xi)= \min_{1\le i\le N}|x_i-\xi|^p, 
$
with $x=(x_1,\ldots,x_N)\in(\R^d)^N$ and $\xi\in \R^d$. Note that, for any~$\xi$, the $i$th entry of this gradient is given by
$$
\big(\nabla_x d_N^p(x,\xi)\big)_i 
= 
p\, |x_i-\xi|^{p-1} \frac{x_i-\xi}{|x_i-\xi|}\,\mathbb I_{[x_i=\text{Proj}_x(\xi)]},
$$
where $\mathbb I_A$ denotes the indicator function of the set $A$, and with the convention $0/0=1$ when $x_i=\xi$. This implies that the~$N$-vector $\nabla_x d_N^p(x,\xi)$ always has exactly one non-zero entry, namely the one corresponding to the point of the grid~$x$ that is closest to~$\xi$. Consequently, at each step~$t$ of the algorithm, only one point of the grid~$X^{t-1}$ will be changed to define the grid~$X^{t}$, namely the point from the grid~$X^{t-1}$ that is closest to~$\xi^{t}$. 

More details about this algorithm can be found in~\citet{Pages03}. For~$p=2$, this is known as the \emph{Competitive Learning Vector Quantization (CLVQ)} algorithm, and is the most commonly used one in quantization. This success is explained by the fact that the convergence results obtained for the CLVQ algorithm are much more satisfactory than for~$p\neq 2$.

\subsection{Convergence results for 
the CLVQ algorithm}
\label{subsubCLVQ}

Here we state several results showing that the grids provided by the CLVQ algorithm 
converge to optimal grids as the number of iterations~$t$ goes to infinity. 


We start with the univariate case ($d=1$). Assume that the support of~$P_X$ is compact and let its convex hull~$C$ be~$[a,b]$. Write $F_N^{+}:=\{x=(x_1,\dots,x_N) : a<x_1<\cdots<x_N<b\}$ for the set of $N$-grids on~$C$ involving  pairwise distinct points stored in ascending order, and let $\bar F_N^{+}$ be its closure; see~\citet{Pages98}. Denote by~$D_N^{2,P_X}(x)=\int_C\min_{1\le i\le N} |x_i-w|^2 P_X(dw)$ the (squared) ${\rm L}_2$-norm quantization error associated with a given grid~$x=(x_1,\dots,x_N)\in F_N^{+}$.

\begin{theorem}[\citealp{Pages98}, Th.~27]\label{conv_d=1}
In the univariate setup above, we have the following. 
\begin{enumerate}[(i)]
\item Assume that $P_X$ is absolutely continuous with a density~$f:[a,b]\to\R^+$ that is positive on~$(a,b)$, and assume either that $f$ is strictly log-concave or that it is log-concave with $f(a+)-f(b-)>0$. Then $x\mapsto D_N^{2,P_X}(x)$ has a unique minimizer~$x^*$ in~$\bar F_N^{+}$, which coincides with the unique solution of $\nabla D_N^{2,P_X}(x)=0$ in~$\bar F_N^{+}$
(when $P_X$ is the uniform over~$[0,1]$,
the optimal grid is~$x^*=\left(a+\frac{2k-1}{2N}(b-a)\right)_{1\le k\le N})$.
\item  
Irrespective of the initial grid~$X^0\in F^{+}_N$, every trajectory~$(X^0,X^1,X^2,\ldots)$ of the CLVQ algorithm is \mbox{a.s.} such that $X^t\in F^{+}_N$ for all~$t$.
%
If $P_X$ is absolutely continuous and if there are finitely many grids~$x(\in\bar F_N^{+})$ such that $\nabla D_N^{2,P_X}(x)=0$, then
$
X^t\xrightarrow{\text{a.s}} x^*
$
as~$t\to\infty$,
with
$\nabla D_N^{2,P_X}(x^*) = 0.
$
\end{enumerate}
\end{theorem}
\vspace{2mm}

Part~(i) of the result provides a particular family of distributions for which the optimal grid is unique (recall that existence always holds). Beyond stating that trajectories of the CLVQ algorithm live in~$F_N^{+}$ (with grids that therefore stays of size~$N$), Part~(ii) of the result provides mild conditions under which the algorithm almost surely provides a limiting grid that is a critical point of the quantization error, hence, under the assumptions of Part~(i), is optimal. 

Unfortunately, the picture is less clear for the multivariate case~($d>1$). While it is still so that the grid~$X^t$ will have pairwise distinct components for any~$t$, some of the components of the limiting grid~$x^*$, if any, may coincide. 
\begin{itemize}
\item[(a)] If, parallel to the univariate case, this does not happen, then the \mbox{a.s.} convergence of~$X^t$ to a critical point of the quantization error~$D_N^{2,P_X}(\cdot)$ can be established under the assumption that~$P_X$ has a bounded density with a compact and convex support. 
\item[(b)] Otherwise, no convergence results are available; the only optimality results that can then be obtained relate to approximations involving grids of size~$k<N$, where~$k$ is the number of distinct components in the limiting grid~$x^*$, which is quite different from the original $N$-quantization problem considered initially.
\end{itemize}
The interested reader may refer to~\citet{Pages98} for details.
For practical purposes, though, one should not worry to much, as all numerical exercices we conducted were compatible with case (a) (with increasing~$t$, the smallest distance between two components of~$X^t$ always seemed to stabilize rather than decreasing to zero).  


\section{Conditional quantiles through optimal quantization}\label{sect_approx}

Let us come back to the regression setup involving a scalar response~$Y$ and a $d$-dimensional vector of covariates~$X$, and consider the conditional quantile functions~$q_\alpha(\cdot)$ in~(\ref{condquant}). It is well-known that
\begin{equation}
\label{eq_def2_cond_qu}
q_\alpha(x)=\arg\textstyle{\min_{a\in\R}}\,{\rm E}\big[\rho_\alpha(Y-a)|X=x\big]
,
\end{equation}
where 
$z\mapsto \rho_\alpha(z)
=-(1-\alpha)z\mathbb I_{[z<0]}+\alpha z\mathbb I_{[z\ge 0]}
= z\big(\alpha-\mathbb I_{[z<0]}\big)$ is the so-called \emph{check function}. As we now explain, this allows to use optimal quantization to approximate conditional quantiles. 

To do so, fix~$p\geq 1$ such that $\|X\|_p<\infty$. Then, for any positive integer~$N$, one may consider the approximation
\begin{equation}
\label{quantizedquantile}
\widetilde{q}^{N}_\alpha(x)
=
\arg\min_{a\in \R}  
{\rm E}\big[\rho_\alpha(Y-a)|\quant{X}=\tilde{x}\big],
\end{equation}
where $\quant{X}$ and $\tilde x$ are the projections of $X$ and $x$ respectively onto an $L_p$-optimal $N$-grid. Since $\quant{X}-X$ goes to zero as~$N\to\infty$, one may expect that~$\widetilde{q}^{N}_\alpha(x)$ provides a better and better approximation of~$q_\alpha(x)$ as~$N$ increases. The main goal of this section is to quantify the quality of this approximation. 

We will need the following assumptions.
\vspace{3mm}

{\sc Assumption~(A)}
(i) 
The random vector~$(X,Y)$ is generated through $Y=m(X,\varepsilon)$, where the $d$-dimensional covariate vector~$X$ and the error~$\varepsilon$  are mutually independent;
(ii) 
the link function $(x,z)\mapsto m(x,z)$ is of the form $m_1(x) + m_2(x) z$, where the functions $m_1(\cdot):\R^d\to\R$ and $m_2(\cdot):\R^d\to\R^+_0$ are Lipschitz functions;
(iii)
$\|X\|_p<\infty$ and $\|\varepsilon\|_p<\infty$; 
(iv) 
the distribution of $X$ does not charge any hyperplane. 
\vspace{3mm}

Note that Assumption~(A)(ii)-(iii) directly implies that there exists $C>0$ such that the link function $m(\cdot,\varepsilon)$ of the model above satisfies 
\begin{equation}
\label{lipc}
\forall u, v \in\R^d, \ \|m(u,\varepsilon) - m(v, \varepsilon)\|_p \le C |u-v|.
\end{equation}
The resulting Lipschitz constant --- that is, the smallest real number~$C$ for which~(\ref{lipc}) holds --- is $[m]_{\mathrm{Lip}}= [m_1]_{\mathrm{Lip}} +[m_2]_{\mathrm{Lip}} \| \varepsilon \|_p$, where $[m_1]_{\mathrm{Lip}}$ and $[m_2]_{\mathrm{Lip}}$ are the corresponding  Lipschitz constants of~$m_1$ and~$m_2$, respectively. 
\vspace{3mm}

{\sc Assumption (B)}  (i) The support $\suppX$ of~$P_X$ is compact; 
(ii) $\varepsilon$ admits a continuous density~$f^\varepsilon:\R\to\R^+_0$ (with respect to the Lebesgue measure on~$\R$).
\vspace{3mm}

To obtain rates of convergence, we will need the following reinforcement of Assumption~(A). 
\vspace{-2mm}

{\sc Assumption~(A$^\prime$)}
Same as Assumption~(A), but with (iii) replaced by (iii)$^\prime$ there exists $\delta>0$ such that $\|X\|_{p+\delta}<\infty$, and $\|\varepsilon\|_p<\infty$.
\vspace{3mm}


We can then prove the following result (see the Appendix for the proof).


\begin{theorem}
\label{rate_conv_q_X}
Fix~$\alpha\in(0,1)$. Then (i) under Assumptions~(A)-(B), 
\vspace{2mm}
$$
\| \widetilde q_\alpha^{N}(X) - q_\alpha(X)\|_p\le 
2\,
\sqrt{
\max\Big(\frac{\alpha}{1-\alpha},\frac{1-\alpha}{\alpha} \Big)
} 
\,
[m]_{\mathrm{Lip}}^{1/2}
\left\|L^N(X)\right\|_p^{1/2}
\|X-\quant{X}\|^{1/2}_p,
\vspace{2mm}
$$
for $N$ sufficiently large, where $(L^N(X))$ is a sequence of $X$-measurable random variables that is bounded in $L_p$;
(ii) under Assumptions~(A$^\prime$)-(B),
$$
\| \widetilde q_\alpha^{N}(X) - q_\alpha(X)\|_p = O(N^{-1/2d}),
$$
as $N \to \infty$.
\end{theorem}
\vspace{1mm}


%

Of course, fixed-$x$ consistency results are also quite appealing in quantile regression. Such a result is provided in the following theorem (see the Appendix for the proof). 

\begin{theorem}
\label{conv_q_X_x}
Fix $\alpha\in(0,1)$. Then, under Assumptions~(A)-(B), 
$$
\sup_{x\in \suppX}
\big|
\widetilde q_\alpha^{N}(x) 
-
q_\alpha(x)
\big|
\to
0,
$$
as $N \to \infty$.
\end{theorem}
\vspace{2mm}

Unlike in Theorem~\ref{rate_conv_q_X}, Theorem~\ref{conv_q_X_x} does not provide any rate of convergence. This is a consequence of the fact that, while the convergence of $\quant{X}$ towards $X$ can be shown to imply the convergence of~$\tilde{x}$ towards $x$ for each fixed~$x$, it does not seem possible to show that the rate of convergence in the fixed-$x$ convergence is inherited from the convergence involving~$X$.


\section{Quantized conditional quantile estimators}
\label{sect_estim}

\subsection{The proposed estimators and their consistency}
\label{soussect_estim}

Consider now the problem of estimating the conditional quantile~$q_\alpha(x)$ on the basis of independent copies $(X_1,Y_1), \dots, (X_n,Y_n)$ of~$(X,Y)$. For any~$N(<n)$, the approximation~$\tildeq$ in~(\ref{quantizedquantile}) leads to an estimator~$\widehat q^{N,n}_{\alpha}(x)$ of the conditional quantile~$q_\alpha(x)$, through the following two steps~: 
\begin{itemize}
\item[(S1)]
First, the CLVQ algorithm from Section~\ref{subsubsect_algo} is applied to perform quantization in~$X$. For this purpose, (i)
the initial grid~$X^0$ is obtained by sampling randomly among the $X_i$'s without replacement, and with the constraint\footnote{If the $X_i$'s are \mbox{i.i.d.} with a common density~$f$, sampling without replacement among the $X_i$'s of course implies that this constraint will be met with probability one. One often needs to impose it, however, in real-data examples (due to the possible presence of ties) or when performing bootstrap (see later).} that the same $x$-values cannot be picked more than once; (ii) $n$ iterations are performed, based on~$\xi^t=X_t$, $t=1,\ldots,n$. We write~$\hat \gamma^{N,n}=(\hat x_{1}^{N,n},\dots,\hat x^{N,n}_{N})$ for the resulting grid and~$\widehat{X}^{N,n}=\text{Proj}_{\hat\gamma^{N,n}}(X)$ for the corresponding (empirical) quantization of~$X$; to make the notation less heavy, we will stress dependence on~$n$ in these quantities only when it is necessary.  
\item[(S2)] Second, the approximation 
$
\tildeq
=
\arg\textstyle{\min_{a}}{\rm E}[\rho_\alpha(Y-a)|\quant{X}=\tilde{x}]
$
is then estimated by
$$
\hatq
=
\arg\textstyle{\min_{a}}
\sum_{i=1}^n 
\rho_\alpha(Y_i -a) 
\,
\mathbb I_{[\widehat X^{N}_i=\hat x^N]},
$$
where~$\widehat X^{N}_i=\widehat X^{N,n}_{i}=\text{Proj}_{\hat\gamma^{N,n}}(X_i)$ and $\hat x^{N}=\hat x^{N,n}=\text{Proj}_{\hat\gamma^{N,n}}(x)$. Of course, $\hatq$, in practice, is simply evaluated as the sample $\alpha$-quantile of the~$Y_i$'s whose corresponding~$\widehat{X}^{N}_i$ is equal to~$\hat x^{N}$.  
\end{itemize}

Note that the number of iterations is equal to the sample size~$n$ at hand, so that it is expected that only moderate-to-large~$n$ will provide reasonable approximations of optimal $N$-grids.

For fixed~$N$ (and~$x$), the convergence in probability of~$\hatq$ to~$\tildeq$ as~$n\to\infty$ can be obtained by making use of the convergence results for the stochastic gradient algorithm discussed in Section~\ref{subsect_quantif}. In order to do so, we need to restrict to~$p=2$ (that is, to the CLVQ algorithm) and to adopt the following assumption. 
\vspace{3mm}

{\sc Assumption~(C)}
$P_X$ is absolutely continuous with respect to the Lebesgue measure on~$\mathbb{R}^d$.
\vspace{-3mm}

We then have the following result. 

%


\begin{theorem}
\label{consistth}
Fix $\alpha \in (0,1)$, $x\in \suppX$ and $N\in \mathbb N_0$. Then, under Assumptions (A), (B)(i), and~(C), we have that, as $n\to\infty$, 
$$
|\hatq - \tildeq|\to 0,
$$
in probability, provided that quantization is based on~$p=2$.
\end{theorem}
\vspace{2mm}

In the previous section, we showed that, as~$N\to \infty$, $\tildeq-q_\alpha(x)$ goes to zero almost surely, hence in probability. Theorem~\ref{consistth} then suggests that  $\hatq - q_\alpha(x)$ might go to zero in probability as both~$n$ and~$N$ go to infinity in some appropriate way. Obtaining such a double asymptotic results, however, is extremely delicate, since, to the best of our knowledge, all convergence results for the stochastic gradient algorithm in the literature are as $n\to \infty$ with~$N$ fixed.

\subsection{Numerical example and bootstrap modification}

For the sake of illustration, we evaluated the estimator~$\hatq$ for~$N=10,25$, and $40$, in a sample of $n=300$ mutually independent observations~$(X_i,Y_i)$ obtained from the model 
\begin{equation}
\label{modelillustr}
Y= \frac 1 5 X^3 + \varepsilon
,
\end{equation}
 where $X=6Z-3$, with $Z\sim \rm {Beta}(0.3,0.3)$, and~$\varepsilon\sim \mathcal N(0,1)$ are independent. The left panels in Figure~\ref{n=300,N=10,25,50} plot the corresponding quantiles curves~$x\mapsto \hatq$ for~$\alpha=0.05, 0.25, 0.5, 0.75$, and $0.95$ (actually, these curves were evaluated only at 300 equispaced points in~$(-3,3)$). 
It is seen that the number of quantizers~$N$ used has an important impact on the curves. For small~$N$, the curves are not smooth and show a large bias. For large~$N$, bias is smaller but the variability is large. One should keep in mind that, for large~$N$, the grid provided by the CLVQ algorithm poorly approximates the corresponding optimal $N$-grid, since a fixed number~$n$ of iterations are used in the algorithm (that should not be too small compared to $N$).

Smoother quantile curves can be obtained from the bootstrap, through the following procedure. For some integer~$B$, we first generate  grids~$\hat \gamma^{N,n}_b$, $b=1,\ldots,B$ (each of size~$N$), as follows from the CLVQ algorithm: first, we sample~$N$ observations with replacement from the initial sample~$X_1,\ldots,X_n$ to generate an initial grid~$X_b^0$ (with the same constraint as in~(S1) above that the~$N$ values obtained are pairwise different); second, we perform iterations based on $\xi^t_b$, $t=1,\ldots,n$, that are similarly obtained by sampling with replacement from~$X_1,\ldots,X_n$ (this time, without any constraint). This allows to consider the bootstrap estimators 
$$
\barq=\frac{1}{B}\sum_{b=1}^B \hat q^{(b)}_{\alpha}(x),
$$
where $\hat q^{(b)}_{\alpha}(x)=\hat q^{(b),N,n}_{\alpha}(x)$ is obtained by performing~(S2) based on the original sample~$(X_i,Y_i)$, $i=1,\ldots,n$, and the grid~$\hat \gamma^{N,n}_b$. Bootstrapping, thus, focuses on the construction of the grids. 

The right panels of Figure~\ref{n=300,N=10,25,50} plot the resulting bootstrapped quantile curves~$x\mapsto \barq$ for the same values of~$\alpha$ and~$N$ as in the original (non-bootstrapped) versions. Bootstrapping clearly smooths all curves, 
and
moreover significantly reduces boundary effects for small~$N$. These advantages require to take~$B$ large enough. But of course, very large values of~$B$ should be avoided in order to keep the computational burden under control.

\begin{figure}
\begin{center}
$
\begin{array}{cc}
\includegraphics[width=7.33cm]{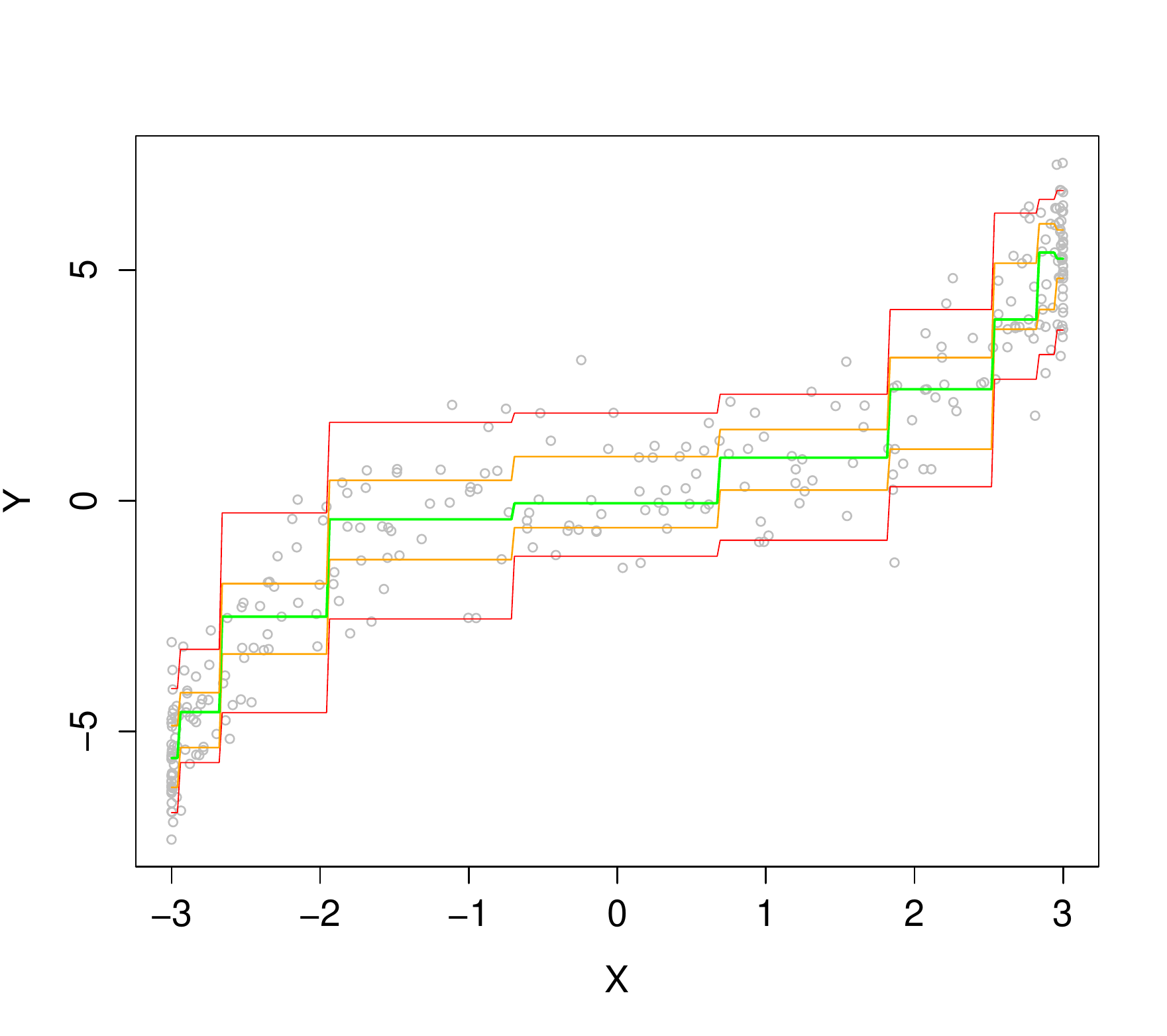}
&
\hspace{2mm}
\includegraphics[width=7.33cm]{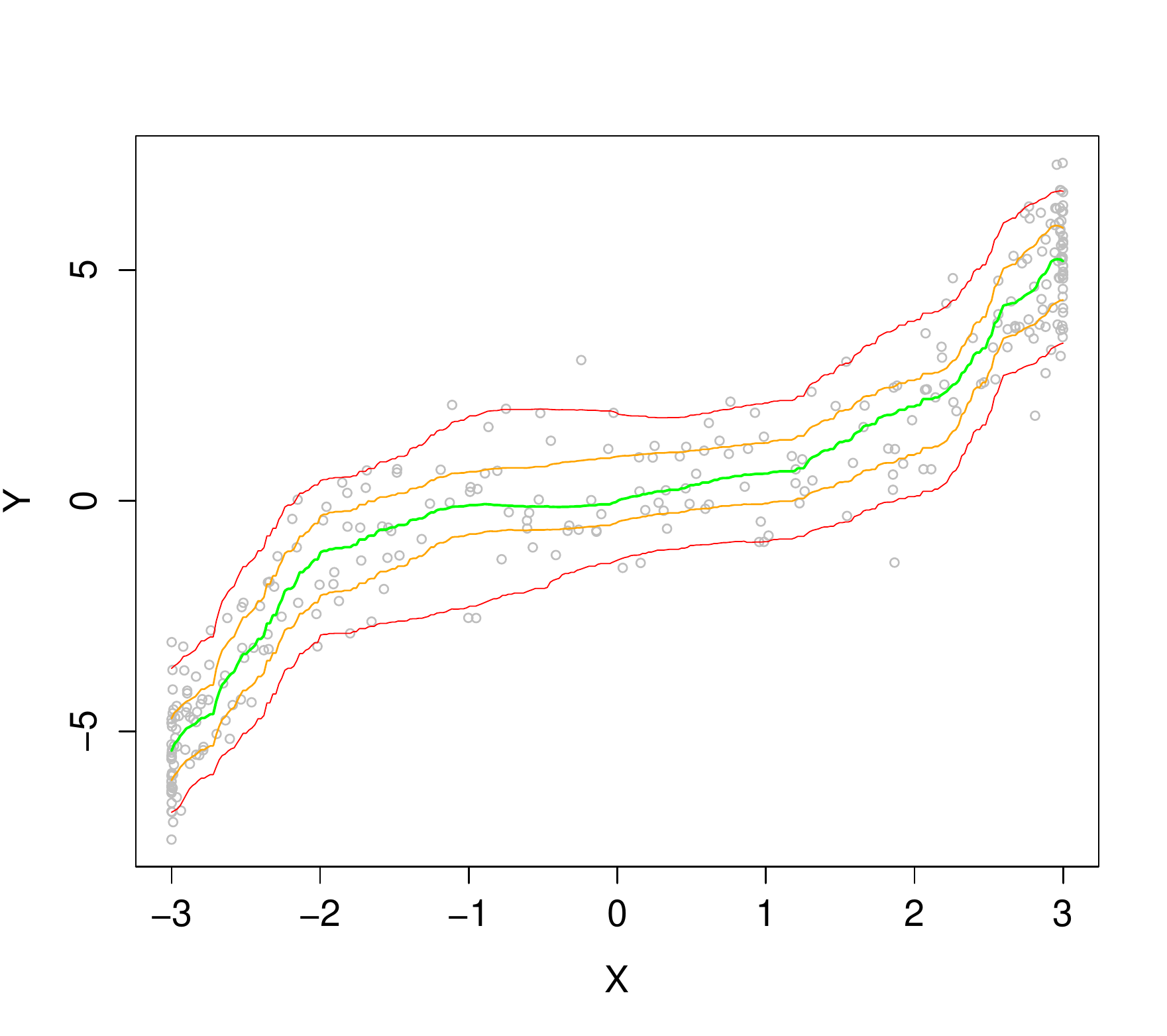}
\\[2mm]
\includegraphics[width=7.33cm]{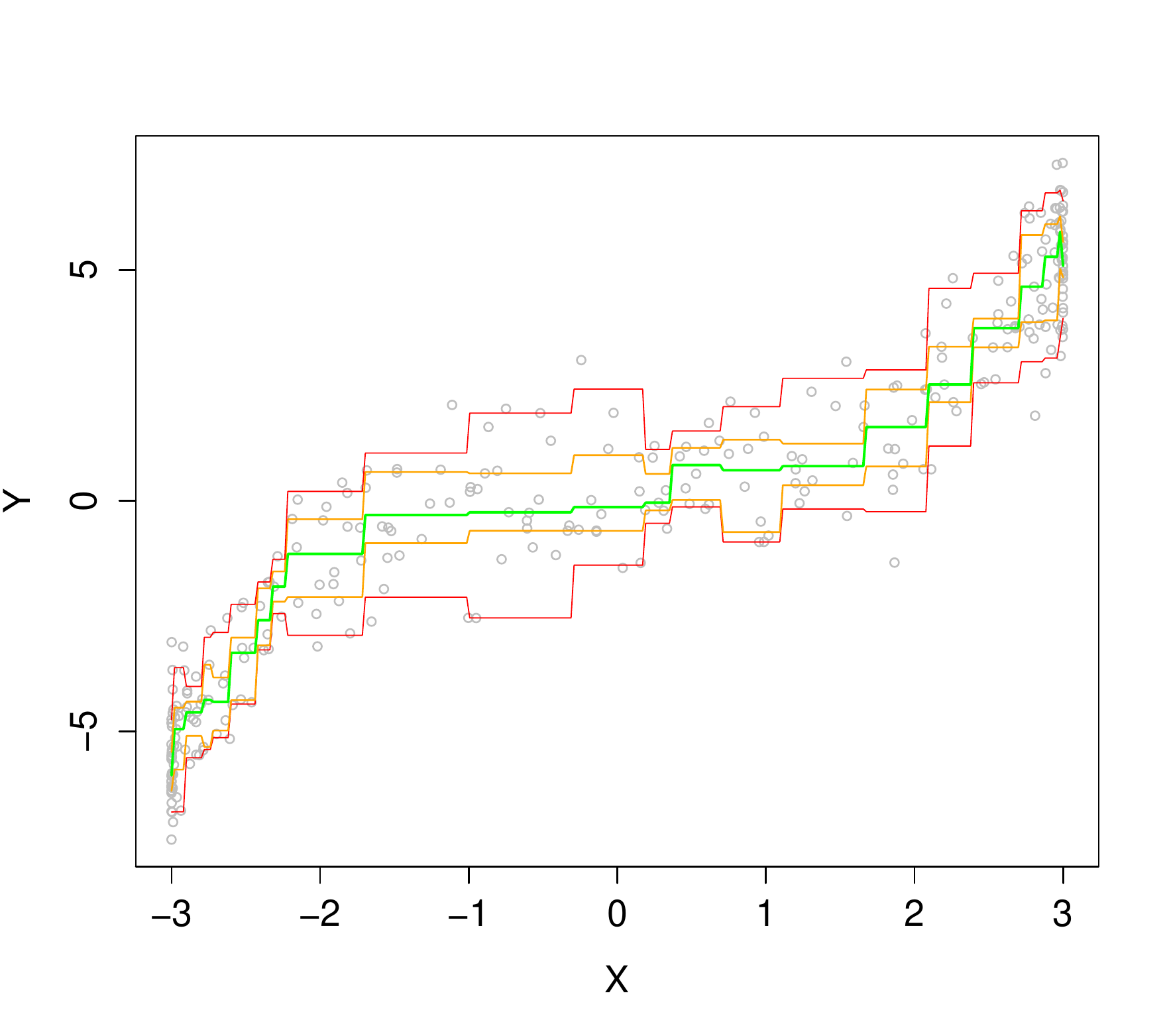}
&
\hspace{2mm}
\includegraphics[width=7.33cm]{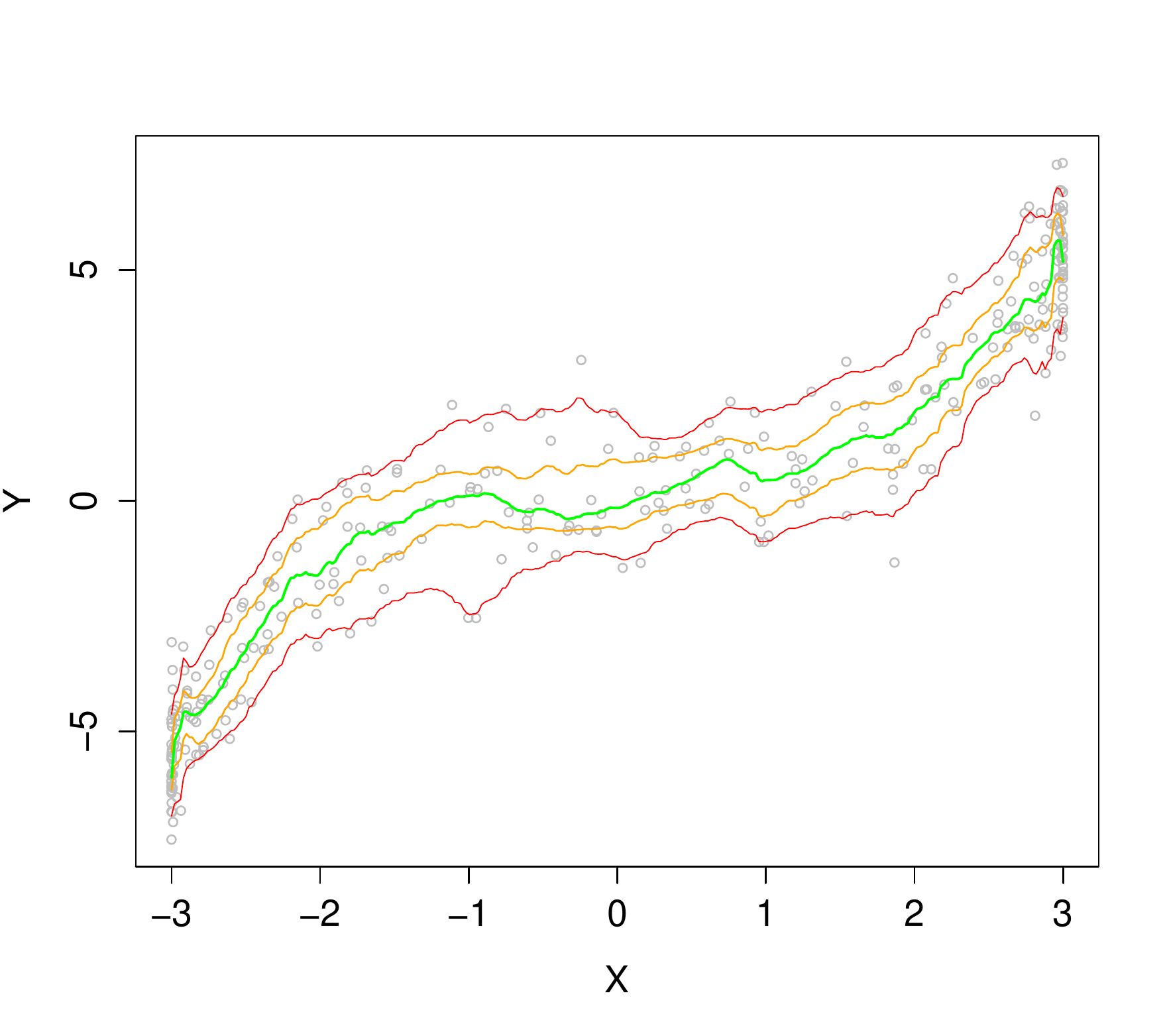}
\\[2mm]
\includegraphics[width=7.33cm]{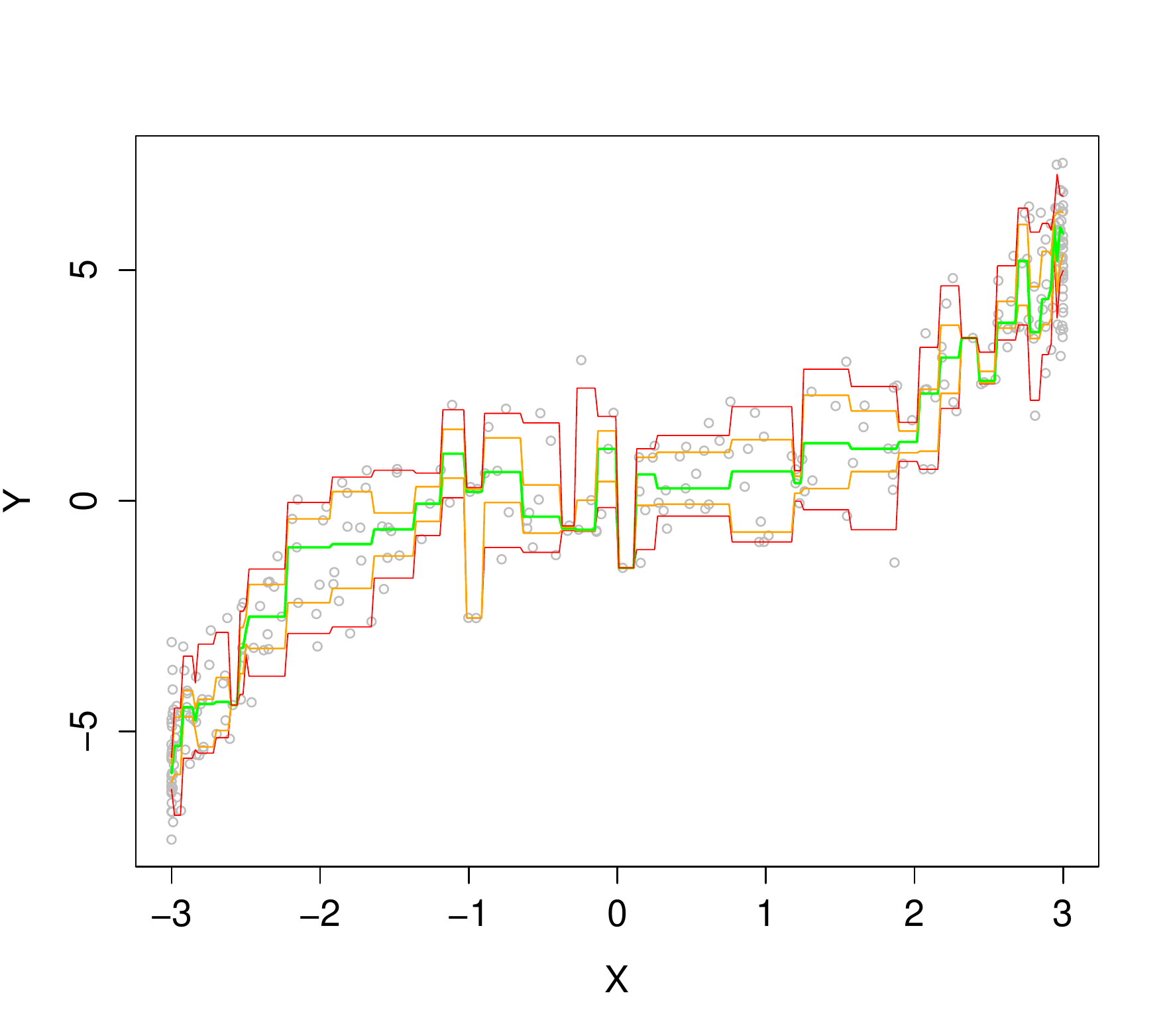}
&
\hspace{2mm}
\includegraphics[width=7.33cm]{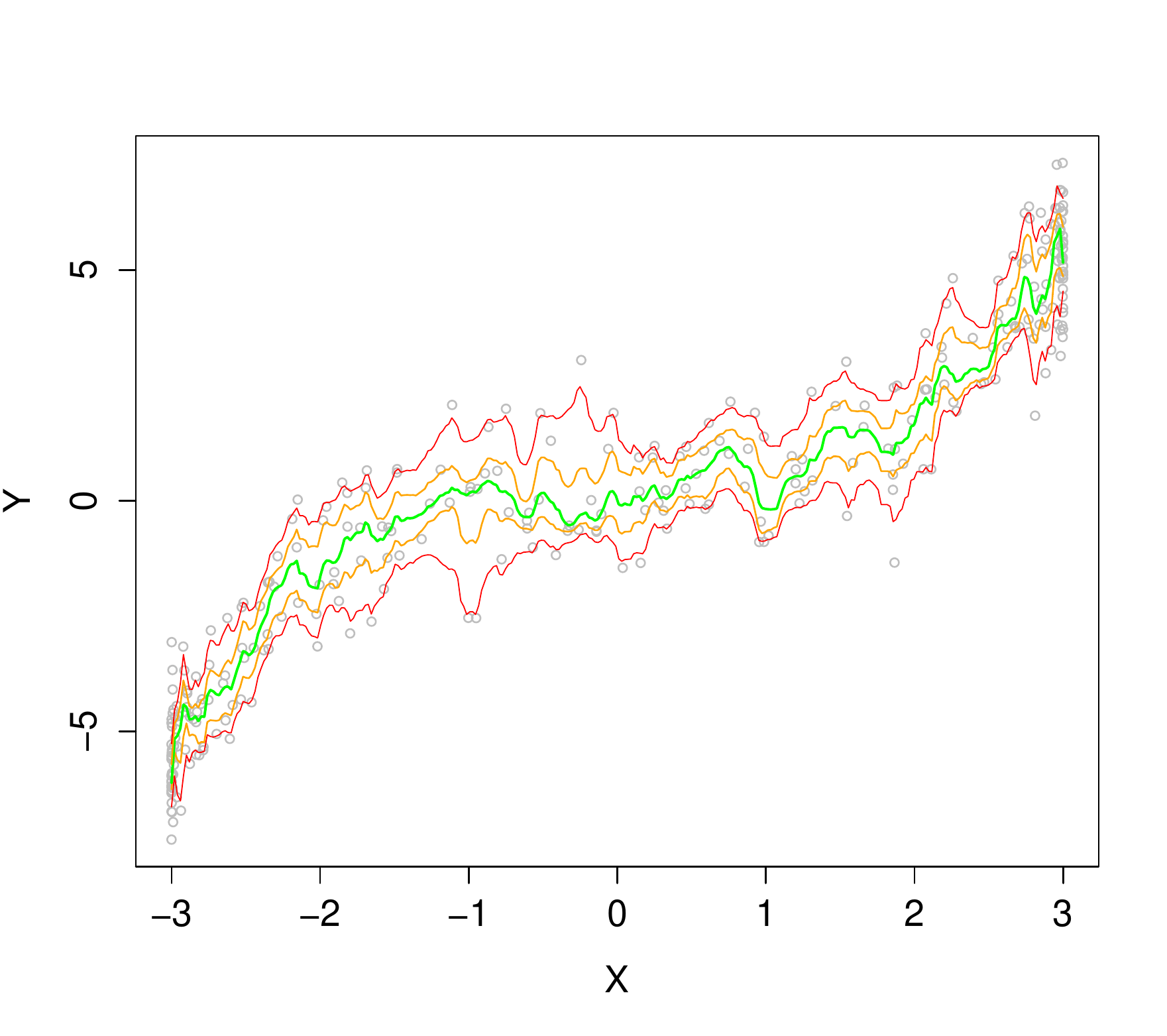}
\end{array}
$
\end{center}
\caption{\footnotesize{Estimated conditional quantiles curves $x\mapsto \hatq$ (left) and their bootstrapped counterparts $x\mapsto\barq$ for $B=50$ (right), based on~$N=10$ (top), $N=25$ (middle), and $N=50$ (bottom).
The sample size is~$n=300$, and the quantiles levels considered are~$\alpha$=0.05, 0.25, 0.5, 0.75, and 0.95. See~(\ref{modelillustr}) for the data generating model.}
\label{n=300,N=10,25,50}}
\end{figure}

\begin{figure}
\begin{center}
$
\begin{array}{cc}
\includegraphics[width=7.33cm]{n=300,beta0,3,B=50,N=25}
&
\hspace{2mm}
\includegraphics[width=7.33cm]{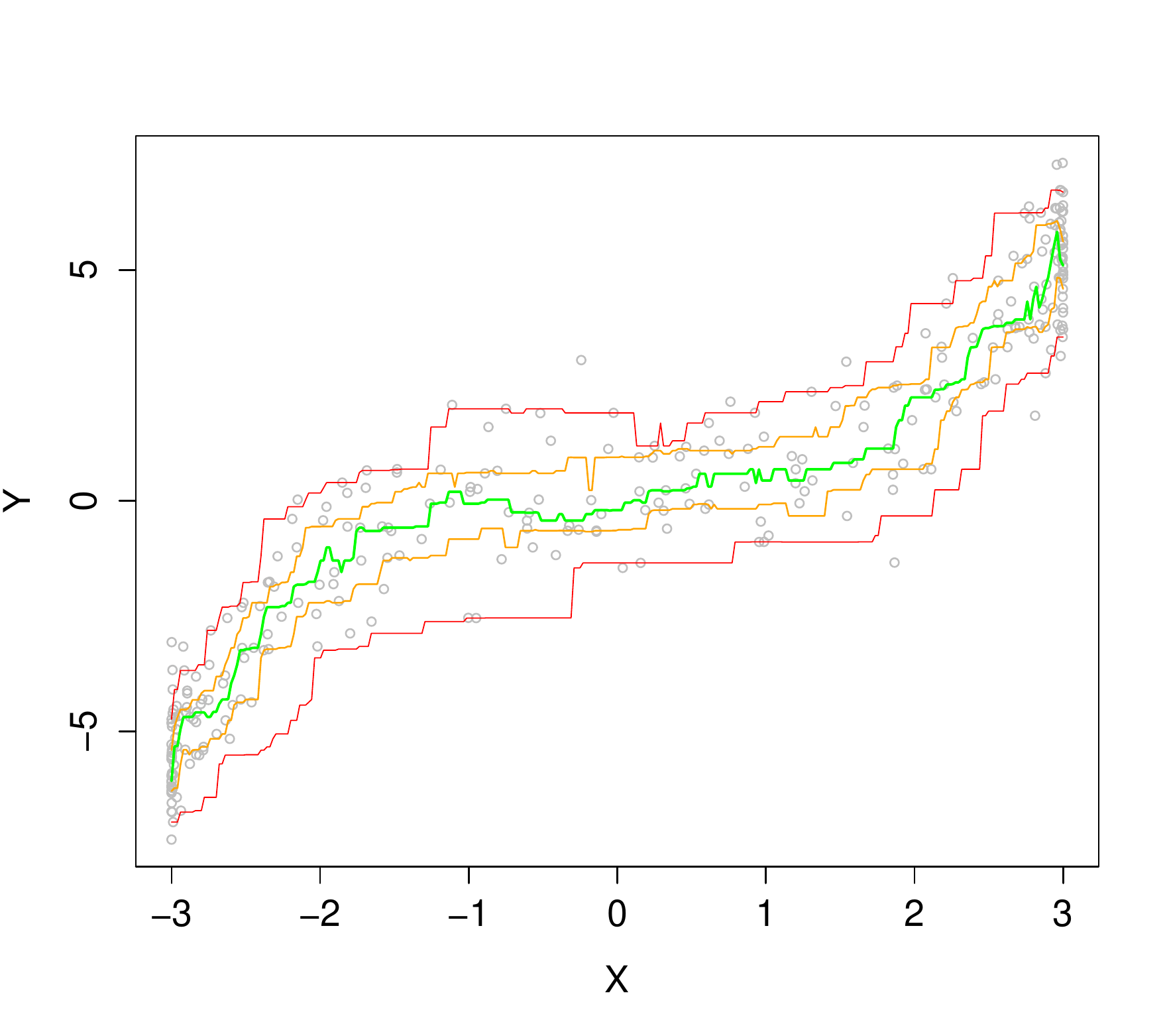}
\\[2mm]
\includegraphics[width=7.33cm]{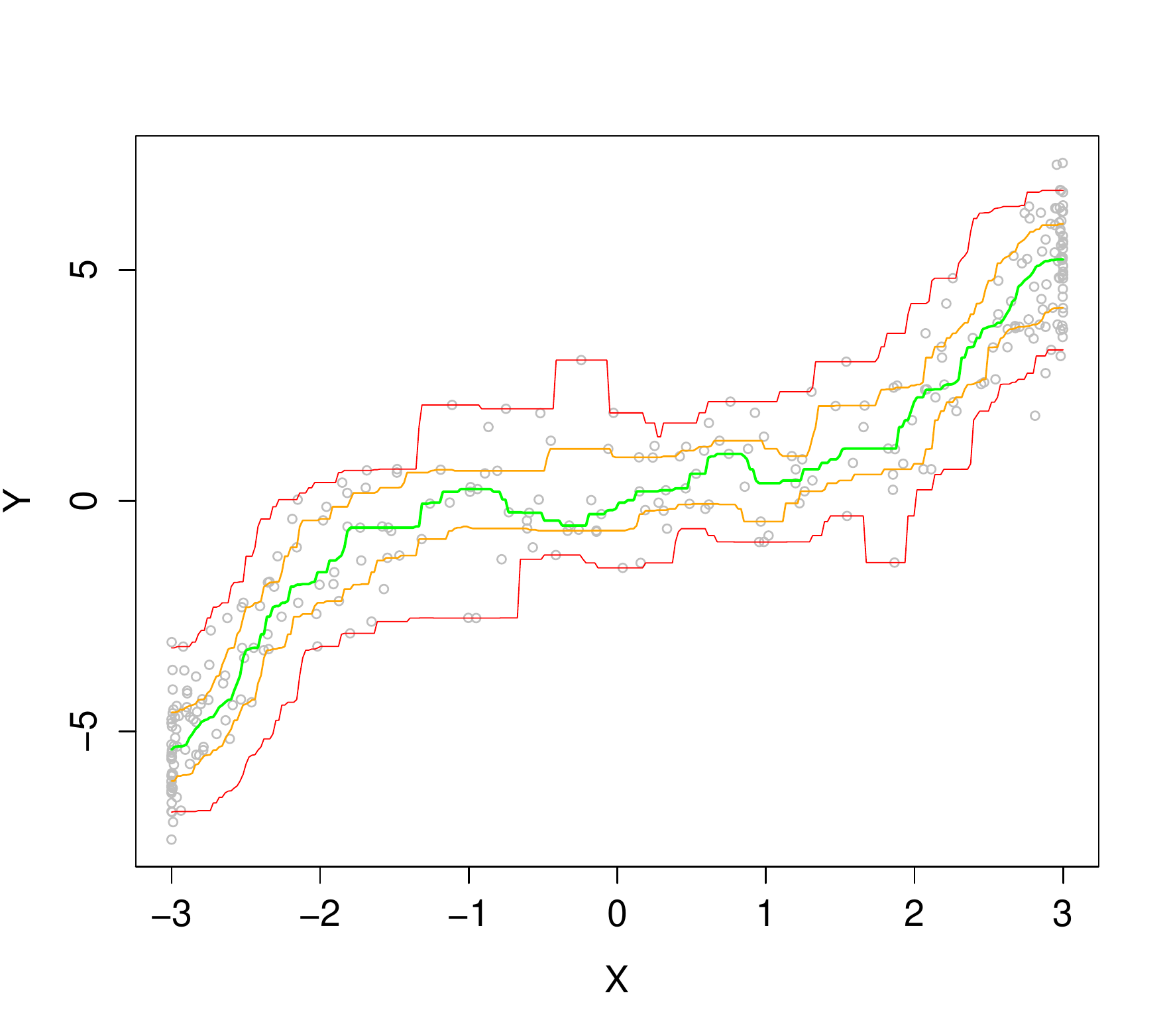}
&
\hspace{2mm}
\includegraphics[width=7.33cm]{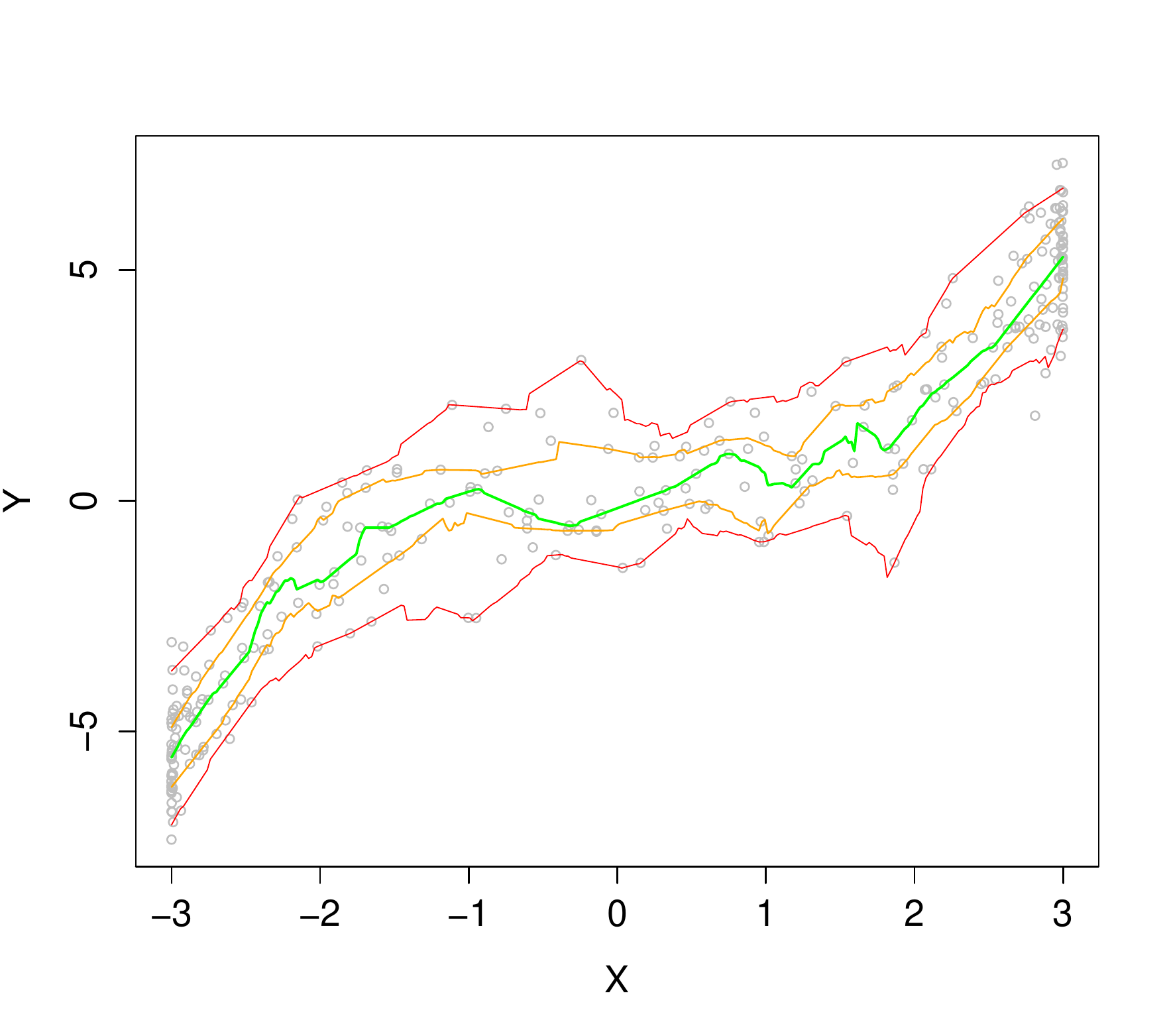}
\end{array}
$
\end{center}
\caption{\footnotesize{The proposed estimated conditional quantiles curves $x\mapsto \barq$ for $B=50$ (top left), their nearest-neighbor competitors
 from \citet{BhattachGango} (top right), and their local constant 
and local linear 
competitors from \citet{YuJones2}  (bottom left and bottom right, respectively). The sample is the same as in Figure~\ref{n=300,N=10,25,50}.}
\label{n=300,all_est}}
\end{figure}

The main competitors to the proposed quantization-based estimators are the nearest-neighbor estimators (as those of \citealp{BhattachGango}) and the local constant and local linear estimators from \citet{YuJones2}. 
For the sake of comparison, we plot those competitors in Figure~\ref{n=300,all_est}, jointly with our bootstrapped estimator (based on $B=50$), on the same sample used in  Figure~\ref{n=300,N=10,25,50}. For each estimator, the smoothing parameters involved 
%
were selected in an automatic way. For nearest-neighbor estimators, the number~$k$ of neighbors to consider was chosen to minimize the MSE
$$
\mathrm{MSE}(k) = \frac{1}{\mathcal N_x} \sum_{i=1}^{\mathcal N_x} \, \big(\widehat q_\alpha^k(x_i)-q_\alpha(x_i)\big)^2,
$$
where $\widehat q_\alpha^k(x)$ denotes the nearest-neighbor estima of $q_\alpha(x)$ based on $k$ neighbors and $\{x_1,\dots,,x_{\mathcal N_x} \}$ is an equispaced grid on~$(-3,3)$ (note that this method uses the population conditional quantiles, hence is infeasible in practice). Local constant/linear estimators (that are based on the Gaussian kernel) involve a bandwidth that, for fixed $\alpha$, is selected (here in a genuinely data-driven way) as
$$
h_\alpha = \frac{\alpha(1-\alpha)h_{\rm mean}}{\phi(\Phi^{-1}(\alpha))^2},
$$
where $\phi$ and $\Phi$ are respectively the standard normal density and distribution functions. Here, $h_{\rm mean}$ corresponds to the optimal bandwidth for regression mean estimation and is chosen through cross-validation; see  \citet{YuJones2}. As for the quantization-based estimators, the number of quantizers~$N$ is selected through a data-driven method that is defined and investigated in the companion paper~\citet{Chaetal2014b}.

Since $Z$ is generated according to a beta distribution with (equal) parameter values that are smaller than one, the  $X_i$'s are less dense in the middle of the interval (-3,3) than close to its boundaries. The local constant and linear estimators seem to suffer from this fact, as the corresponding quantiles curves are significantly less smooth in the middle than at the boundaries. The nearest-neighbor estimator and our quantization-based estimator provide better estimations  in the middle part, with an advantage for our estimator that  appears to be smoother (despite the fact that, unlike for the nearest-neighbor estimator, the corresponding smoothing parameter value is chosen in a totally data-driven way).

\section{Final comments}
\label{sect_concl} 

In this paper, we presented a new method to estimate nonparametrically conditional quantile curves of~$Y$ given~$X$. The main idea is to use optimal quantization as an alternative to more standard localization techniques such as kernel or nearest-neighbor methods. Our construction consists in replacing~$X$ in the definition of conditional quantiles with a quantized version of~$X$. We showed that this provides a valid approximation of conditional quantiles. In the empirical case, this approximation leads to new estimators of conditional quantiles, that, as we proved, are (for fixed~$N$) consistent to their population counterparts.   

We illustrated the behavior of these estimators on a numerical example, and showed that~$N$ essentially behaves as a smoothing parameter, parallel to the bandwidth and the number of neighbors to be used for kernel and nearest-neighbor methods, respectively. Of course, extensive simulations are needed to compare the performances of the proposed estimators with these classical competitors. This is achieved in the companion paper~\citet{Chaetal2014b}, where it is shown that quantization-based estimators tend to dominate their competitors in terms of MSEs as soon as $X$ is not uniformly distributed. There, a method to choose empirically the smoothing parameter~$N$ is also developed and investigated.   


\appendix

\section{Proofs of Section~\ref{sect_approx}}

In this section, we prove Theorems~\ref{rate_conv_q_X} and~\ref{conv_q_X_x}, which requires to establish several lemmas. We first introduce some notation. Let  $G_a(x)={\rm E}[\rho_\alpha(Y-a)|X=x]$ and denote the corresponding quantized quantity by $\widetilde{G}_a(\tilde{x})={\rm E}[\rho_\alpha(Y-a)|\quant{X}=\tilde{x}]$. Since $\widetilde{q}=\widetilde{q}^{N}_\alpha(x)=\arg\min_{a\in \R} \widetilde{G}_a(\tilde{x})$ and $q=q_\alpha(x)=\arg\min_{a\in \R} G_a(x)$, it is natural to try to control the distance between $\widetilde{G}_a(\tilde{x})$ and $G_a(x)$. This is achieved in Lemma~\ref{lemme_conv_q_x}, the proof of which requires the following preliminary result.

\begin{lemma}
\label{lemlip}
Fix~$\alpha\in(0,1)$ and $a\in\R$. Then, under Assumption~(A), (i) $\rho_\alpha:\mathbb{R}\to\mathbb{R}$ is Lipschitz, with Lipschitz constant $[\rho_\alpha]_{\mathrm{Lip}} = \max(\alpha, 1-\alpha)$, and (ii) $G_a:\mathbb{R}^d\to\mathbb{R}$ is Lipschitz, with Lipschitz constant $[G_a]_{\mathrm{Lip}}=\max(\alpha,1-\alpha)
[m]_{\mathrm{Lip}}$.
\end{lemma}

\begin{proof}[Proof of Lemma~\ref{lemlip}]
%
Since (i) is a trivial calculus exercise, we only prove~(ii). To do so, note that, for any~$u,v \in \R^d$,
\begin{align*}
|G_a(u) - G_a(v)| 
&= 
\big|
{\rm E}[\rho_\alpha(Y-a)|X=u] - {\rm E}[\rho_\alpha(Y-a)|X=v]
\big|
\\
&=
\big
|{\rm E}[\rho_\alpha(m(X,\varepsilon)-a)|X=u] - {\rm E}[\rho_\alpha(m(X,\varepsilon)-a)|X=v]|
,
\end{align*}
so that the independence of~$X$ and $\varepsilon$ entails
$$
|G_a(u) - G_a(v)| 
 \leq
{\rm E}[|\rho_\alpha(m(u,\varepsilon)-a)-\rho_\alpha(m(v,\varepsilon)-a)|]
\le  
[\rho_\alpha]_{\text{Lip}} [m]_{\mathrm{Lip}}|u-v|
,
$$
where the last inequality follows from Part~(i) of the result and equation~\eqref{lipc}.
\end{proof}

We still need the following lemma to prove Theorem~\ref{conv_q_X_x}.

\begin{lemma}
\label{lemme_conv_q_x}
Fix $\alpha \in (0,1)$ and $x\in \suppX$. For any integer~$N$, let $\tilde{x} =\tilde{x}^N =\mathrm{Proj}_{\gamma^N}(x)$ and $C_{x}=C^N_{x}=\{z\in\suppX : {\rm Proj}_{\gamma^N}(z)=\tilde{x}\}$. 
Then,  under Assumptions~(A)-(B), \\
(i) $\sup_{x \in \suppX} |x-\tilde{x}|
\to
0$ as $N\to\infty$;
\\
(ii) $\sup_{x\in\suppX} R(C_{x})
\to
0$ as $N\to\infty$, where we let $R(C_{x}):= \sup_{z\in C_{x}}|z-\tilde{x}|$;
\\
(iii) $\sup_{x\in\suppX}\sup_{a\in\R}|\widetilde G_a(\tilde{x}) - G_a(x)|  
\to
0$ as $N \to \infty$;
\\
(iv) $\sup_{x\in\suppX}|\min_{a\in \R} \widetilde G_a(\tilde{x}) - \min_{a\in\R}G_a(x)| 
\to
0$ as $N \to \infty$. 
\end{lemma}
\vspace{2mm}



\begin{proof}[Proof of Lemma~\ref{lemme_conv_q_x}]
(i) Assume by contradiction that there exists $\varepsilon>0$ such that, for infinitely many~$N$ (for~$N\in \mathcal{N}(\varepsilon)$, say), we have $\sup_{x \in \suppX}|\tilde{x}^N-x| >\varepsilon$. For any such value of~$N$, one can pick $x\in\suppX$ (that may depend on~$N$) with $|\tilde{x}^N-x|>\varepsilon$. No point of the optimal grid~$\gamma^N$
belongs to the ball~$\mathcal B(x,\varepsilon)=\{z\in\R^d:|z-x|<\varepsilon\}$, which implies that, for all $z\in \mathcal B(x,\varepsilon/2)$, $|\tilde z^N-z|>\varepsilon/2$, 
where $\tilde z$ is the projection of $z$ onto~$\gamma^N$. Therefore, for~$N\in\mathcal{N}(\varepsilon)$, 
\begin{eqnarray}
\lefteqn{
\|X^{\gamma^N}-X\|_p^p    
=
 \int_{\text{Supp}(P_X)} |\tilde{z}^N-z|^p\, dP_X(z)
\geq 
}
\nonumber
\\[2mm]
& & 
\int_{\mathcal B(x,\varepsilon/2)} |\tilde{z}^N-z|^p\, dP_X(z)
>
\left(\frac{\varepsilon}{2}\right)^p
\inf_{y\in \suppX} 
P_X\big[\mathcal B(y,\varepsilon/2)\big]
=:\delta_\varepsilon
>0
,
\label{tt}
\end{eqnarray}
where the last inequality follows from the fact that $y\mapsto P_X\big[\mathcal B(y,\varepsilon/2)\big]$ is a continuous function taking only strictly positive values on the compact set~$\suppX$. 
Since the cardinality of~$\mathcal{N}(\varepsilon)$ is infinite, (\ref{tt}) prevents $\|X^{\gamma^N}-X\|_p$ to go to zero as~$N\to\infty$, a contradiction. 
\\
(ii) Since $\suppX$ is compact, we have that, for any~$x$, the radius of the quantization cell~$C_x$ is bounded. Hence, for any~$x$, there exists $x^{*N}\in C_x$ such that 
$$
R(C_{x})= \sup_{z\in C_{x}}|z-\tilde{x}|=|x^{*N}-\tilde{x}|\leq \sup_{x \in \suppX} |x-\tilde{x}|.
$$ 
Hence,
$$
\sup_{x \in \suppX} R(C_{x})\leq \sup_{x \in \suppX} |x-\tilde{x}|.
$$
The result then follows from Part~(i).
\\
(iii) Fix $a\in\R$. First note that, since $[\quant{X}=\tilde{x}]$ is equivalent to $[X\in C_{x}]$, one has   
$$
|
{\rm E}[\rho_\alpha(Y-a)|\quant{X}=\tilde{x}] 
-
{\rm E}[\rho_\alpha(Y-a)|X=\tilde{x}]
|
\le
\sup_{z \in C_{x}} 
|
{\rm E}[\rho_\alpha(Y-a)|X=z] 
-
{\rm E}[\rho_\alpha(Y-a)|X=\tilde{x}]
|
.
$$
Therefore,
\begin{align*}
\lefteqn{
\hspace{5mm}
|\widetilde G_a(\tilde{x}) - G_a(x)| 
=
|
{\rm E}[\rho_\alpha(Y-a)|\quant{X}=\tilde{x}] 
-
{\rm E}[\rho_\alpha(Y-a)|X=x]
|
}
\\[2mm]
&\le
|
{\rm E}[\rho_\alpha(Y-a)|\quant{X}=\tilde{x}] 
-
{\rm E}[\rho_\alpha(Y-a)|X=\tilde{x}]
|
+
|
{\rm E}[\rho_\alpha(Y-a)|X=\tilde{x}]
-
{\rm E}[\rho_\alpha(Y-a)|X=x]
|
\\[2mm]
&\le
2 \sup_{z \in C_{x}} 
|
{\rm E}[\rho_\alpha(Y-a)|X=z] 
-
{\rm E}[\rho_\alpha(Y-a)|X=\tilde{x}]
|.
\end{align*}
Using the independence between~$X$ and~$\varepsilon$ and the Lipschitz properties of~$\rho_\alpha$ and~$m$ then yields
\begin{align}
\lefteqn{
\hspace{-10mm}
|\widetilde G_a(\tilde{x}) - G_a(x)| 
\leq 
2
\sup_{z \in C_{x}} 
|
{\rm E}[\rho_\alpha(m(z, \varepsilon)-a) 
-
{\rm E}[\rho_\alpha(m(\tilde{x},\varepsilon)-a)]
|
}
\nonumber
\\[2mm]
&\le
2\max(\alpha,1-\alpha)[m]_{\mathrm{Lip}}
\sup_{z \in C_{x}} 
|
z - \tilde{x}
|
=
2\max(\alpha,1-\alpha)[m]_{\mathrm{Lip}}
R(C_{x})
.
\label{part_cell}
\end{align}
Hence,
$$
\sup_{x \in \suppX} 
\sup_{a\in\R} 
|\widetilde G_a(\tilde{x}) - G_a(x)| 
\le 
2\max(\alpha,1-\alpha)[m]_{\mathrm{Lip}}
\sup_{x \in \suppX} R(C_{x}).
$$
The results then follows from Part~(ii).

(iv) Letting $\mathbb I_+ = \mathbb I_{[\min_{a\in \R}\widetilde G_a(\tilde{x}) \ge \min_{a\in\R}G_a(x)]}$, we have
\begin{align*}
\lefteqn{
\hspace{-10mm}
| \min_{a\in \R}\widetilde G_a(\tilde{x}) - \min_{a\in\R}G_a(x)| \mathbb I_+
= \big(\widetilde G_{\tildeq}(\tilde{x}) - G_{q_\alpha(x)}(x)\big)\mathbb I_+
}
\\[3mm]
\hspace{10mm}
&
\le 
\big(\widetilde G_{q_\alpha(x)}(\tilde{x}) - G_{q_\alpha(x)}(x)\big))\mathbb I_+
\le \sup_{a\in\R} |\widetilde G_a(\tilde{x}) - G_a(x)|\mathbb I_+
.
\end{align*}
Proceeding similarly with $\mathbb I_- = \mathbb I_{[\min_{a\in \R}\widetilde G_a(\tilde{x}) < \min_{a\in\R}G_a(x)]}$, this yields
\begin{align*}
\lefteqn{
\hspace{-10mm}
| \min_{a\in \R}\widetilde G_a(\tilde{x}) - \min_{a\in\R}G_a(x)| \mathbb I_-
= \big(G_{q_\alpha(x)}(x) - \widetilde G_{\tildeq}(\tilde{x})\big)\mathbb I_-
}
\\[3mm]
\hspace{10mm}
&
\le \big(G_{\tildeq}(x)- \widetilde G_{\tildeq}(\tilde{x})\big)\mathbb I_-
\le \sup_{a\in\R} |\widetilde G_a(\tilde{x}) - G_a(x)|\mathbb I_-
,
\end{align*}
 so that $| \min_{a\in \R}\widetilde G_a(\tilde{x}) - \min_{a\in\R}G_a(x)| \le \sup_{a\in\R} |\widetilde G_a(\tilde{x}) - G_a(x)|$.
The result therefore follows from Part~(iii).
\end{proof}
\vspace{2mm}


We can now prove Theorem~\ref{conv_q_X_x}.

\begin{proof}[Proof of Theorem~\ref{conv_q_X_x}]
First note that, for any~$x\in\suppX$, 
\begin{align*}
|G_{\tildeq}(x) - G_{q_\alpha(x)}(x)| 
&\le
 |G_{\tildeq}(x) - \widetilde G_{\tildeq}(\tilde{x})|
 + 
| \widetilde G_{\tildeq}(\tilde{x}) - G_{q_\alpha(x)}(x)|
\\[2mm]
&\le
\sup_{a\in\R} |G_a(x) - \widetilde G_a(\tilde{x})|
 + 
| \min_{a\in \R}\widetilde G_a(\tilde{x}) - \min_{a\in\R}G_a(x)|
\\[2mm]
&\le
\sup_{x\in\suppX}\sup_{a\in\R} |G_a(x) - \widetilde G_a(\tilde{x})|
 + 
\sup_{x\in\suppX}| \min_{a\in \R}\widetilde G_a(\tilde{x}) - \min_{a\in\R}G_a(x)|.
\end{align*}
Therefore, Lemma~\ref{lemme_conv_q_x}(iii)-(iv) readily implies that, as $N\to\infty$,  
\begin{equation}
\label{you}
\sup_{x\in\suppX}|G_{\tildeq}(x) - G_{q_\alpha(x)}(x)| \to 0.
\end{equation}
%

Now, let~$\tilde{N}$ be such that, for any~$N\geq \tilde{N}$, we have $|G_{\tildeq}(x) - G_{q_\alpha(x)}(x)|\leq 1$ for all~$x\in\suppX$. As we will show later in this proof, this implies that there exists~$M$ such that
\begin{equation}
\label{toshowcomp}
|\tildeq-q_\alpha(x)|\leq M, 
\end{equation}
for all $x\in\suppX$ and~$N\geq \tilde{N}$. 

One can easily check that, for any~$x\in\suppX$, $a\mapsto G_a(x)$ is twice continuously differentiable, with derivatives  
$$
\frac{dG_a(x)}{da} = F^{\varepsilon}\bigg(\frac{a-m_1(x)}{m_2(x)}\bigg)-\alpha
\ \textrm{ and } \ \
\frac{d^2G_a(x)}{da^2} =
\frac{1}{m_2(x)}
\, f^{\varepsilon}\bigg(\frac{a-m_1(x)}{m_2(x)}\bigg).
$$
Consequently, performing a second-order expansion about $a=q_\alpha(x)$ provides
$$
G_{\tildeq}(x) - G_{q_\alpha(x)}(x) = 
\frac{1}{m_2(x)}\,
f^{\varepsilon}\bigg(\frac{q^N_{\alpha *}(x)-m_1(x)}{m_2(x)}\bigg) 
(\tildeq-q_\alpha(x))^2,
$$
for some $q^N_{\alpha *}(x)$ between  $\tildeq$ and $q_\alpha(x)$. Therefore, 
\begin{equation}
\label{sidmq}
\sup_{x\in\suppX}(\tildeq-q_\alpha(x))^2
\leq
\frac{\sup_{x\in\suppX} m_2(x)}{\inf_{x\in\suppX}f^{\varepsilon}\Big(\frac{q^N_{\alpha *}(x)-m_1(x)}{m_2(x)}\Big)}\sup_{x\in\suppX}|G_{\tildeq}(x) - G_{q_\alpha(x)}(x)|.
\end{equation}
Since $m_2(\cdot)$ is a continuous function defined over the compact set $\suppX$, we have 
\begin{equation}
\label{dgsj}
\sup_{x\in\suppX} m_2(x) \leq C_a
\end{equation}
 for some constant $C_a$. 
Using~(\ref{toshowcomp}) and the fact that $q_\alpha(\cdot)$, $m_1(\cdot)$, and $m_2(\cdot)$  are continuous functions also defined over this compact set (with $m_2(\cdot)$ taking strictly positive values), we have that, for $N\geq \tilde{N}$,  
$$
\sup_{x\in\suppX}
\frac{|q^N_{\alpha *}(x)-m_1(x)| }{m_2(x)}
\leq \frac{\sup_{x\in\suppX} \big ( |q_\alpha(x)|+|m_1(x)| + M \big)}{\inf_{x\in\suppX} m_2(x)} \leq C_b,
$$
for some constant~$C_b$. Jointly with the continuity of the positive function~$f^\varepsilon(\cdot)$, this implies that the infimum in~(\ref{sidmq}) admits a strictly positive lower bound that is independent of~$N$ (for $N\geq \tilde{N}$).  Using this, (\ref{you}) and~(\ref{dgsj}), we conclude from~(\ref{sidmq}) that
$$
\sup_{x\in\suppX}(\tildeq-q_\alpha(x))^2
\to 0,
$$
as~$N\to \infty$, which was to be proved. 

It remains to show that the claim in~(\ref{toshowcomp}) indeed  holds true. By contradiction, assume that for all~$M$, there exists~$x=x_M$ (and $N\geq \tilde{N}$) such that $|\tildeq-q_\alpha(x)|> M$. The convexity of~$a\mapsto G_a(x)$ and the fact that $|G_{\tildeq}(x) - G_{q_\alpha(x)}(x)|\leq 1$ for all~$x\in\suppX$ (for $N\geq \tilde{N}$) readily implies that, for any~$a$ with $|a-q_\alpha(x)|\leq M$,
$$
G_a(x)\leq G_{q_\alpha(x)}(x) + \frac{|a-q_\alpha(x)|}{M}.  
$$
Since this holds for all arbitrarily large~$M$, the convexity of~$G_a(x)$ implies that  
$$
f^{\varepsilon}\left(\frac{q_\alpha(x)+1-m_1(x)}{m_2(x)}\right)
=
\frac{d^2G_a(x)}{da^2}\Bigg|_{a=q_\alpha(x)+1}\to 0
$$
as~$M\to\infty$. The same argument as above, however, shows that $\inf_{x\in\suppX}f^{\varepsilon}\left(\frac{q_\alpha(x)+1-m_1(x)}{m_2(x)}\right)>0$, a contradiction. 
\end{proof}
\vspace{2mm}


The proof of Theorem~\ref{rate_conv_q_X} requires the following three lemmas.

\begin{lemma}
\label{rate_conv_E_X}
Fix~$\alpha\in(0,1)$. Then (i) under Assumptions~(A)-(B), we have 
\begin{equation}\label{bound_prop}
\big\|
\textstyle{\sup_{a}} \big| \widetilde{G}_a(\quant{X}) - G_a(X)\big|
\big\|_p
\le  
2
\max(\alpha,1-\alpha)
[m]_{\mathrm{Lip}}
\big\|
\quant{X} - X
\big\|_p
\,
;
\end{equation}
(ii) under Assumptions~(A$^\prime$)-(B), we have that
$$
\big\|
\textstyle{\sup_{a}} \big| \widetilde{G}_a(\quant{X}) - G_a(X)\big|
\big\|_p
= 
O\big(N^{-1/d}\,\big)
,
$$
as~$N\to \infty$.
\end{lemma}
\vspace{2mm}

\begin{lemma}
\label{rate_conv_min_X}
Fix~$\alpha\in(0,1)$. Then (i) under Assumptions~(A)-(B), we have
$
\big\| \widetilde G_{\tilde{q}}(\quant{X}) - G_{q}(X) \big\|_p
\leq 2
\max(\alpha,1-\alpha)[m]_{\mathrm{Lip}}
\big\|
\quant{X} - X
\big\|_p
$;
(ii) under Assumptions~(A$^\prime$)-(B), $\big\| \widetilde G_{\tilde{q}}(\quant{X}) - G_{q}(X) \big\|_p
=O(N^{-1/d})$ as~$N\to\infty$. 
\end{lemma}
\vspace{2mm}


\begin{lemma}
\label{lemmeK}
Let Assumption (B) hold. For any~$x\in\suppX$, let $L(x) = 1/f^{Y|X=x}(q_\alpha(x))$ and  $L^N(x) = 1/f^{Y|X=x}(c^N(x))$, where $c^N(x)$ is the infimum of all $c$'s between $\tilde q^N_\alpha(x)$ and $q_\alpha(x)$ for which 
\begin{equation}
\int_{\min(q_\alpha(x),\tilde q^N_\alpha(x))}^{\max(q_\alpha(x),\tilde q^N_\alpha(x))} f^{Y|X=x}(y) \,dy= f^{Y|X=x}(c)\, |\tilde q^N_\alpha(x)-q_\alpha(x)|
\label{definc}
\end{equation}
(existence follows from the mean value theorem). Then
$
\|L^N(X)\|_p \to \left\|L(X)\right\|_p
$
as $N\to\infty$.
\end{lemma}
\vspace{0mm}

\begin{proof}[Proof of Lemma~\ref{rate_conv_E_X}]
Part~(ii) of the result readily follows from Part~(i) and Corollary~\ref{rate_conv_qu}, so that we may focus on the proof of Part~(i). Note that $G_a(\quant{X})$ stands for the conditional expectation of $\rho_\alpha(Y-a)$ given that $X=\quant{X}$, which is different from ${\rm E}[\rho_\alpha(Y-a)|\quant{X}]$. For any~$a$, 
\begin{eqnarray*}
\big|
\widetilde{G}_a(\quant{X}) - G_a(X)
\big|
&\leq &
\big|
\widetilde{G}_a(\quant{X})-G_a(\quant{X})
\big|
+
\big|
G_a(\quant{X}) -G_a(X)
\big|
\\[2mm]
&\leq &
\textstyle{\sup_{a}}
\big|
\widetilde{G}_a(\quant{X})-G_a(\quant{X})
\big|
+
\textstyle{\sup_{a}}
\big|
G_a(\quant{X}) -G_a(X)
\big|
\end{eqnarray*}
almost surely, so that 
$$
\sup_a
\big|
\widetilde{G}_a(\quant{X}) - G_a(X)
\big|
\leq
\textstyle{\sup_{a}}
\big|
\widetilde{G}_a(\quant{X})-G_a(\quant{X})
\big|
+
\textstyle{\sup_{a}}
\big|
G_a(\quant{X}) -G_a(X)
\big|
$$
almost surely. The triangular inequality then yields
\begin{equation}
\|
\textstyle{\sup_{a}} \big| \widetilde{G}_a(\quant{X}) - G_a(X) \big|
\big\|_p
\le 
\big\|
\textstyle{\sup_{a}} \big|  \widetilde{G}_a(\quant{X})-G_a(\quant{X}) \big|
\big\|_p  
+ 
\big\|
\textstyle{\sup_{a}} \big| G_a(\quant{X}) -G_a(X) \big|
\big\|_p
.
\label{eq_depart}
\end{equation}
Since $\quant{X}$ is $X$-measurable, 
we have that 
$$
\widetilde{G}_a(\quant{X}) 
=
{\rm E}\big[\rho_\alpha(Y-a)|\quant{X}\big] = {\rm E}\big[{\rm E}[\rho_\alpha(Y-a)|X]|\quant{X}\big] ={\rm E}\big[G_a(X)|\quant{X}\big],
$$
which yields
\begin{eqnarray*}
\textstyle{\sup_{a}}
\big| 
\widetilde{G}_a(\quant{X}) - G_a(\quant{X}) 
\big|
&=& 
\textstyle{\sup_{a}}
\big| 
{\rm E} [ G_a(X) - G_a(\quant{X})|\quant{X}]
\big|
\\[2mm]
&\le &
{\rm E} \Big[ \textstyle{\sup_{a}}
 | G_a(X) - G_a(\quant{X}) | 
\,
\Big|
\quant{X}
\Big]
,
\end{eqnarray*}
almost surely. From Jensen's inequality, we then obtain 
$$
\big\|
\textstyle{\sup_{a}} \big|  \widetilde{G}_a(\quant{X})-G_a(\quant{X}) \big|
\big\|_p  
\leq
\big\|
\textstyle{\sup_{a}} \big|  G_a(X)-G_a(\quant{X}) \big|
\big\|_p  
.
$$
Substituting in~(\ref{eq_depart}) and using Lemma~\ref{lemlip}(ii) yields 
\begin{eqnarray*}
\|
\textstyle{\sup_{a}} \big| \widetilde{G}_a(\quant{X}) - G_a(X) \big|
\big\|_p
&\le &
2
\big\|
\textstyle{\sup_{a}} \big|  G_a(X)-G_a(\quant{X}) \big|
\big\|_p  
\\[1mm]
& \le &
2
\max(\alpha,1-\alpha)
 [m]_{\mathrm{Lip}}
\big\|
\quant{X} - X
\big\|_p
,
\end{eqnarray*}
which establishes the result.
\end{proof}


\begin{proof}[Proof of Lemma~\ref{rate_conv_min_X}]
Letting $\mathbb{I}_+=\mathbb{I}_{[\widetilde G_{\tilde{q}}(\quant{X}) \geq  G_{q}(X)]}$, note that
\begin{eqnarray*}
\lefteqn{
\hspace{-10mm}
| \widetilde G_{\tilde{q}}(\quant{X}) - G_{q}(X) |\,
\mathbb{I}_+
\le
\big(
\widetilde{G}_{\tilde{q}}(\quant{X})-G_{q}(X)
\big)
\mathbb{I}_+
}
\\[2mm]
& &
\hspace{10mm}
\le
\big(
\widetilde{G}_{q}(\quant{X})-G_{q}(X)
\big) 
\mathbb{I}_+
 \le 
\Big(
\sup_{a}
\big|
\widetilde{G}_{a}(\quant{X})-G_{a}(X)
\big|
\Big)
\mathbb{I}_+,
\end{eqnarray*}
almost surely. Similarly, with $\mathbb{I}_-=\mathbb{I}_{[\widetilde G_{\tilde{q}}(\quant{X}) <  G_{q}(X)]}$, we have
\begin{eqnarray*}
\lefteqn{
\hspace{-10mm}
| \widetilde G_{\tilde{q}}(\quant{X}) - G_{q}(X) |\,
\mathbb{I}_-
\le
\big(
G_{q}(X)-\widetilde{G}_{\tilde{q}}(\quant{X})
\big) 
\mathbb{I}_-
}
\\[2mm]
& &
\hspace{10mm}
\le
\big(
G_{\tilde{q}}(X)-\widetilde{G}_{\tilde{q}}(\quant{X})
\big) 
\mathbb{I}_-
 \le 
\Big(
\sup_{a}
\big|
\widetilde{G}_{a}(\quant{X})-G_{a}(X)
\big|
\Big)
\mathbb{I}_-
,
\end{eqnarray*}
almost surely, so that 
$
| \widetilde G_{\tilde{q}}(\quant{X}) - G_{q}(X) |
 \le 
\sup_{a}
\big|
\widetilde{G}_{a}(\quant{X})-G_{a}(X)
\big|
,
$
almost surely. The result then directly follows from Lemma~\ref{rate_conv_E_X}. 
\end{proof}

\begin{proof}[Proof of Lemma~\ref{lemmeK}]
We have to prove that, as $N\to\infty$,
\begin{equation}
\label{toprovDCT}
\int_{\suppX }
\frac{1}
{\big(f^{Y|X=x}(c^N(x))\big)^p}
\,dP_X(x)
\to 
\int_{\suppX }
\frac{1}
{\big(f^{Y|X=x}(q_\alpha(x))\big)^p}
\,dP_X(x).
\end{equation}
First note that Assumption~(B) ensures that,  for any~$x\in\suppX$, 
$$
y\mapsto f^{Y|X=x}(y)
=
\frac{1}{m_2(x)}\,
f^{\varepsilon}\bigg(\frac{y - m_1(x)}{m_2(x)}\bigg) 
$$
is continuous. Therefore, Theorem~\ref{conv_q_X_x}, which entails that~$c^N(x) \to q_\alpha(x)$ for any~$x$ as $N\to\infty$, implies that   
$$
\frac{1}
{\big(f^{Y|X=x}(c^N(x))\big)^p}
\to 
\frac{1}
{\big(f^{Y|X=x}(q_\alpha(x))\big)^p},
$$
still for any~$x$ as $N\to\infty$. To establish~(\ref{toprovDCT}), it is then sufficient --- in view of Lebesgue's dominated convergence theorem --- to prove that, for any~$x$ and any (sufficiently large)~$N$,
\begin{equation}
\label{gsll}
\frac{1}
{f^{Y|X=x}(c^N(x))}
=
\frac{m_2(x)}{f^{\varepsilon}\Big(\frac{c^N(x) - m_1(x)}{m_2(x)}\Big)}\,
\leq C
\end{equation}
for some constant~$C$ that does not depend on~$N$. 

To show~(\ref{gsll}), note that Theorem~\ref{conv_q_X_x} and the continuity of $m_1(\cdot)$ and $m_2(\cdot)$  (with $m_2(\cdot)$ taking strictly positive values) over the compact set~$\suppX$ entails that, for $N$ sufficiently large,
\begin{eqnarray*}
\bigg| \frac{c^N(x)-m_1(x) }{m_2(x)}\bigg| 
&=&
\bigg| \frac{(q_\alpha(x)-m_1(x))+(c^N(x)-q_\alpha(x)) }{m_2(x)}\bigg| 
\\[2mm]
&=&
\bigg| \varepsilon_\alpha + \frac{c^N(x)-q_\alpha(x) }{m_2(x)}\bigg|
\\[2mm]
&\leq&
| \varepsilon_\alpha | + \frac{|c^N(x)-q_\alpha(x)| }{D_1}
\leq
| \varepsilon_\alpha | + \frac{1}{D_1}=D_2,
\end{eqnarray*}
for some constants~$D_1,D_2$ (that do not depend on~$N$), where~$\varepsilon_\alpha$ denotes the $\alpha$-quantile of~$\varepsilon$. Consequently,~(\ref{gsll}) directly follows from the continuity of $m_2(\cdot)$ and $f^\varepsilon(\cdot):\R\to\R^+_0$.
\end{proof}
\vspace{2mm}


Finally, we prove Theorem~\ref{rate_conv_q_X}.

\begin{proof}[Proof of Theorem~\ref{rate_conv_q_X}]
Throughout the proof, we write $q(x)$ and $\tilde q(x)$ for   $q_\alpha(x)$ and $\tilde q^N_\alpha(x)$, respectively. Similarly, $q$ and $\tilde q$ will stand for $q_\alpha(X)$ and $\tilde q^N_\alpha(X)$, respectively. 

(i) Let first $r, s \in \R$ with $r\leq s$. For all $y\in\R$, one has $1\ge \mathbb I_{[y\le r]} + \mathbb I_{[y>s]}$, which implies
$$
\rho_\alpha(y-r) 
-
\rho_\alpha(y-s)
\ge
 -(1-\alpha)(s-r)\mathbb I_{[y\le r]}
+
\alpha(s-r)\mathbb I_{[y>s]}
,
$$
hence
\begin{eqnarray*}
\big\{
\rho_\alpha(Y-\tilde q)
- 
\rho_\alpha(Y-q)
\big\}
 \mathbb I_{[\tilde q \leq q]}
\ge
\big\{
-(1-\alpha)(q-\tilde q)\mathbb I_{[Y\le \tilde q]}
+ 
\alpha(q-\tilde q)\mathbb I_{[Y>q]}
\big\}
 \mathbb I_{[\tilde q \leq q]}
.
\end{eqnarray*}
Taking expectation conditional on $X$, this gives
\begin{align}
\lefteqn{
\hspace{-10mm}
|G_{\tilde q}(X)-G_{q}(X)|
 \mathbb I_{[\tilde q \leq q]}
=
\big(G_{\tilde q}(X)-G_{q}(X)\big)
 \mathbb I_{[\tilde q \leq q]}
}
\nonumber
\\[2mm]
&
\hspace{10mm}
\ge
 \big\{
-(1-\alpha)(q-\tilde q)P[Y\le \tilde q|X]
+
\alpha(q-\tilde q)P[Y>q|X]
\big\}
 \mathbb I_{[\tilde q \leq q]}
\nonumber
\\[2mm]
&
\hspace{10mm}
=
(1-\alpha)\big(q-\tilde q\big)\big(\alpha - P[Y\le \tilde q |X]\big)
 \mathbb I_{[\tilde q \leq q]}
\nonumber
\\[2mm]
&
\hspace{10mm}
\geq
\min(\alpha,1-\alpha)|\tilde q-q|\big(P[Y\le q |X] - P[Y\le \tilde q |X]\big)
 \mathbb I_{[\tilde q \leq q]}
\nonumber
\\[2mm]
&
\hspace{10mm}
=
\min(\alpha,1-\alpha)|\tilde q-q| P[\min(\tilde q, q)<Y\le \max(\tilde q, q) |X]
 \mathbb I_{[\tilde q \leq q]}
,
\label{toadd1}
\end{align}
almost surely.

Now, let $r, s \in \R$ with $r>s$. For all $y\in\R$, one has $1\ge \mathbb I_{[y\le s]} + \mathbb I_{[y>r]}$, which implies
$$
\rho_\alpha(y-r) 
-
\rho_\alpha(y-s)
\ge
 -(1-\alpha)(s-r)\mathbb I_{[y\le s]}
+
\alpha(s-r)\mathbb I_{[y>r]}
,
$$
hence
\begin{eqnarray*}
\big\{
\rho_\alpha(Y-\tilde q)
- 
\rho_\alpha(Y-q)
\big\}
 \mathbb I_{[\tilde q > q]}
\ge
\big\{
-(1-\alpha)(q-\tilde q)\mathbb I_{[Y\le q]}
+ 
\alpha(q-\tilde q)\mathbb I_{[Y>\tilde q]}
\big\}
 \mathbb I_{[\tilde q > q]}
.
\end{eqnarray*}
Taking expectation conditional on $X$, this gives
\begin{align}
\lefteqn{
\hspace{-2mm}
|G_{\tilde q}(X)-G_{q}(X)|
 \mathbb I_{[\tilde q > q]}
=
\big(G_{\tilde q}(X)-G_{q}(X)\big)
 \mathbb I_{[\tilde q > q]}
}
\nonumber
\\[2mm]
&
\hspace{10mm}
\ge
 \big\{
-(1-\alpha)(q-\tilde q)P[Y\le  q|X]
+
\alpha(q-\tilde q)P[Y>\tilde q|X]
\big\}
 \mathbb I_{[\tilde q > q]}
\nonumber
\\[2mm]
&
\hspace{10mm}
=
\alpha\big(q-\tilde q\big)\big( P[Y>\tilde q |X]-(1-\alpha)\big)
 \mathbb I_{[\tilde q > q]}
\nonumber
\\[2mm]
&
\hspace{10mm}
\geq 
\min(\alpha,1-\alpha)|\tilde q-q|\big( P[Y\le \tilde q |X]-P[Y\le q |X] \big)
 \mathbb I_{[\tilde q > q]}
\nonumber
\\[2mm]
&
\hspace{10mm}
=
\min(\alpha,1-\alpha)|\tilde q-q| P[\min(\tilde q, q)<Y\le \max(\tilde q, q) |X]
 \mathbb I_{[\tilde q > q]}
,
\label{toadd2}
\end{align}
almost surely. 
Adding up~(\ref{toadd1}) and~(\ref{toadd2}) then provides
\begin{equation}
\label{pres}
\big|
G_{\tilde q}(X) - G_{q}(X)
\big|
\ge
 \min(\alpha,1-\alpha)|\tilde q-q|P[\min(\tilde q, q)< Y\le \max(\tilde q, q)|X].
\end{equation}

Now, for any~$x\in\suppX$, 
\begin{eqnarray*}
\lefteqn{
\hspace{-10mm}
P[\min(\tilde q, q)< Y\le \max(\tilde q, q)|X=x]
=
\int_{\min(q(x),\tilde q(x))}^{\max(q(x),\tilde q(x))} f^{Y|X=x}(y) \,
dy
}
\\[2mm]
& & 
\hspace{20mm}
=
f^{Y|X=x}(c^N(x)) \, |\tilde q(x)- q(x)| 
=
\frac{|\tilde q(x)- q(x)|}{L^N(x)} ,
\end{eqnarray*}
where $c^N(x)$ and $L^N(x)$ were defined in Lemma~\ref{lemmeK}, so that
$$
P[\min(\tilde q, q)\le Y<\max(\tilde q, q)|X]=\frac{|\tilde q-q|}{L^N(X)}
$$
almost surely. Plugging into~(\ref{pres}) yields
$$
|\tilde q -q |^2\le \frac{1}{\min(\alpha, 1-\alpha)} L^N(X) |G_{\tilde q}(X)-G_{q}(X)|,
$$
or equivalently, 
$$
|\tilde q - q|^p 
\le \frac{1}{(\min(\alpha, 1-\alpha))^{p/2}} (L^N(X))^{p/2} |G_{\tilde q}(X)-G_{q}(X)|^{p/2}.
$$
Taking expectations, applying Cauchy-Schwarz inequality in the righthand side, then computing $p$th roots, provides
\begin{equation}
\label{presqqq}
\|\tilde q- q\|_p
\le
 \frac{1}{\sqrt{\min(\alpha, 1-\alpha)}}\left\|L^N(X)\right\|_p^{1/2}\|G_{\tilde q}(X) - G_{q}(X)\|_p^{1/2}.
\end{equation}

From Lemmas~\ref{rate_conv_E_X}-\ref{rate_conv_min_X}, we obtain
\begin{eqnarray*}
\big\|G_{\tilde{q}}(X) - G_{q}(X)\big\|_p
&\le&
 \big\|G_{\tilde{q}}(X) - \widetilde G_{\tilde{q}}(\quant{X})\big\|_p+\big\|\widetilde G_{\tilde{q}}(\quant{X}) - G_{q}(X)\big\|_p
\\[2mm]
&\le&
\big\|\textstyle{\sup_{a}}  |G_{a}(X) - \widetilde G_{a}(\quant{X})|\big\|_p
+\big\|\widetilde G_{\tilde{q}}(\quant{X}) - G_{q}(X)\big\|_p
\\[2mm]
&\le&
4
\max(\alpha,1-\alpha)
[m]_{\mathrm{Lip}}
\big\|
\quant{X} - X
\big\|_p
.
\end{eqnarray*}
The result then follows by plugging this into~(\ref{presqqq}) (the boundedness of $L^N(X)$ in $L_p$ is a direct corollary of Lemma~\ref{lemmeK}). 

(ii) The result directly follows from Part~(i) and Corollary~\ref{rate_conv_qu}.
%
\end{proof}
\vspace{2mm}


\section{Proof of Theorem~\ref{consistth}}

Let $\gamma^N=\gamma^N(X)= \{\tilde x_1,\dots, \tilde x_N\}$ be an optimal grid and $\hat \gamma^{N,n}=\hat \gamma^{N,n}(X_1,\ldots,X_n)=(\hat x_{1}^{N,n},\dots,\hat x^{N,n}_{N})$ be the grid provided by the CLVQ algorithm. Throughout this section, we assume the almost sure convergence of the empirical quantization of~$X$ to the population one, that is 
\begin{equation}
\label{empirconv}
\widehat X^{N,n}=\mathrm{Proj}_{\hat\gamma^{N,n}}(X) \xrightarrow[n\to\infty]{\text{a.s.}} 
\mathrm{Proj}_{\gamma^{N}}(X) =\quant{X},
\end{equation}
which is justified by the discussion in Section~\ref{subsubCLVQ}. 

The proof of Theorem~\ref{consistth} then requires Lemmas~\ref{lemme_conv_p_i}-\ref{lemme_conv_est} below.

\begin{lemma}
\label{lemme_conv_p_i}
Let Assumption (C) hold. Fix $N\in \mathbb N_0$ and $x\in\suppX$, and write $\tilde{x}=\mathrm{Proj}_{\gamma^{N}}(x)$ and $\tilde{x}=\mathrm{Proj}_{\hat\gamma^{N,n}}(x)$. Then, with $\widehat X^{N}_i=\mathrm{Proj}_{\hat\gamma^{N,n}}(X_i)$, $i=1,\dots,n$, we have
\\[2mm]
(i)
$\frac{1}{n}\sum_{i=1}^n\,\mathbb I_{[\widehat X_i^{N}=\hat x^N]}
\xrightarrow[n\to\infty]{\text{a.s.}} 
P[\quant{X}=\tilde{x}]$;
\\[2mm]
(ii)
 after possibly reordering the $\tilde{x}_i$'s, 
$\hat x^{N,n}_i \xrightarrow[n\to\infty]{\text{a.s.}} \tilde{x}_i$, $i=1,\dots,N$ (hence, $\hat \gamma^{N,n} \xrightarrow[n\to\infty]{\text{a.s.}} \gamma^N$).
\end{lemma}
\vspace{1mm}

\begin{proof}
Under~(\ref{empirconv}), Part~(i) was shown in \citet{Bally_al} (see also \citet{Pages98}) and Part~(ii) only states the \mbox{a.s.} convergence of the supports~$\hat \gamma^{N,n}$ to $\gamma^N$, which is a necessary condition for the corresponding convergence of random vectors in~(\ref{empirconv}).
\end{proof}

\begin{lemma}
\label{lemme_conv_est}
Fix $\alpha \in (0,1)$, $x\in \suppX$ and $N\in \mathbb N_0$, let $K(\subset \mathbb{R})$ be compact, and define
$$
\widehat G_{a}^{N,n}(\hat x^N)
:= 
 \frac{
\frac{1}{n}\sum_{i=1}^n\rho_\alpha(Y_i - a)
\,
\mathbb I_{[\widehat X_i^{N}=\hat x^N]}}{\frac{1}{n}\sum_{i=1}^n\,\mathbb I_{[\widehat X_i^{N}=\hat x^N]}} 
.
$$
Then, under Assumptions (A) and~(C), 
(i)~
$\sup_{a\in K} |\widehat G^{N,n}_{a} (\hat x^N) - \widetilde G_{a}(\tilde{x})|=o_{\rm P}(1)$ as~$n\to\infty$;
(ii)~
$|\min_{a\in \R} \widehat G^{N,n}_{a} (\hat x^N) - \min_{a\in \R}\widetilde G_a(\tilde{x})|=o_{\rm P}(1)$ as~$n\to\infty$;
(iii)~
$|\widetilde G_{\hatq}-\widetilde G_{\tildeq}|=o_{\rm P}(1)$ as~$n\to\infty$.
\end{lemma}
\vspace{0mm}

\begin{proof}
(i) Since
$$
\widetilde G_{a}(\tilde{x})
=
{\rm E}[\rho_\alpha(Y-a)|\quant{X}=\tilde{x}]
=
\frac{{\rm E}[\rho_\alpha(Y-a)\mathbb{I}_{[\quant{X}=\tilde{x}]}]}{P[\quant{X}=\tilde{x}]},
$$
it is sufficient --- in view of Lemma~\ref{lemme_conv_p_i}(i) --- to prove that, as~$n\to\infty$,  
$$
\sup_{a\in K}
\bigg| 
\frac{1}{n}\sum_{i=1}^n\rho_\alpha(Y_i - a)
\,
\mathbb I_{[\widehat X_i^{N}=\hat x^N]}
-
{\rm E}[\rho_\alpha(Y-a)\mathbb{I}_{[\quant{X}=\tilde{x}]}]
\bigg|
=o_{P}(1).
$$
Of course, it is natural to consider the decomposition
$$
\sup_{a\in K}
\bigg| 
\frac{1}{n}\sum_{i=1}^n\rho_\alpha(Y_i - a)
\,
\mathbb I_{[\widehat X_i^{N}=\hat x^N]}
-
{\rm E}[\rho_\alpha(Y-a)\mathbb{I}_{[\quant{X}=\tilde{x}]}]
\bigg|
\leq
\sup_{a\in K} |T_{a1}| + \sup_{a\in K} |T_{a2}|,
$$
with
$$
T_{a1}
=
\frac{1}{n}
\sum_{i=1}^n\rho_\alpha(Y_i - a)
\,
\big(
\mathbb I_{[\widehat X_i^{N}=\hat x^N]}
-
\mathbb I_{[\widetilde X_i^{N}=\tilde x]}
\big)
$$
and
$$
T_{a2}
=
\frac{1}{n}\sum_{i=1}^n\rho_\alpha(Y_i - a)
\,
\mathbb I_{[\widetilde X_i^{N}=\tilde x]}
-
{\rm E}[\rho_\alpha(Y-a)\mathbb{I}_{[\quant{X}=\tilde{x}]}]
.
$$
Using the fact that $m_1(\cdot)$ and $m_2(\cdot)$ are continuous functions defined over the compact set~$\suppX$, we obtain that, for any~$a\in K$, there exist positive constants~$C_1$ and $C_2$ such that
$$
\rho_\alpha(Y-a)
\leq
\max(\alpha,1-\alpha)
|Y-a|
\leq 
\max(\alpha,1-\alpha)
(| m_1(X)| + |m_2(X)|\, |\varepsilon| + |a|)
\leq 
C_1+C_2|\varepsilon|,
$$
that is in $L_1$ (recall that~$\varepsilon$ is assumed to be in $L_p$, $p=2$), the uniform law of large numbers (see, e.g., Theorem~16(a) in \citealp{Fre1996}) shows that $\sup_{a\in K} |T_{a2}|=o_P(1)$ as $n\to\infty$. 

Turning to $T_{a1}$, consider the set $\mathcal{I}_n=\{i=1,\ldots,n : I_{[\widehat X_i^{N}=\hat x^N]} \neq \mathbb I_{[\widetilde X_i^{N}=\tilde x]} \}$ collecting the indices of observations that are projected on the same point as~$x$ for~$\gamma^N$ but not for $\hat\gamma^{N,n}$ (or vice versa on the same point as~$x$ for~$\hat\gamma^{N,n}$ but not for $\gamma^N$). For any~$a\in K$, we have
\begin{eqnarray*}
\lefteqn{
\hspace{-30mm}
|T_{a1}|
\leq 
\frac{1}{n}
\sum_{i\in \mathcal{I}_n}
\big|\rho_\alpha(Y_i - a)\big|
\leq 
\frac{\max(\alpha,1-\alpha)}{n}
\sum_{i\in \mathcal{I}_n}
(| m_1(X_i)| + |m_2(X_i)|\, |\varepsilon_i| + |a|)
}
\\[2mm] 
& & 
\hspace{7mm}
\leq
\frac{\#\mathcal{I}_n}{n} \times
\frac{1}{\#\mathcal{I}_n}
\sum_{i\in \mathcal{I}_n}
(C_1+ C_2\, |\varepsilon_i|)
=:
S_1\times S_2.
\end{eqnarray*}
Clearly, Lemma~\ref{lemme_conv_p_i}(ii) implies that $\# \mathcal{I}_n/n=o_P(1)$ as~$n\to\infty$, while the independence between $\mathcal{I}_n$ (which is measurable with respect to the $X_i$'s) and the $\varepsilon_i$'s entails that ${\rm E}[S_2]=O(1)$ as~$n\to\infty$, so that $S_2$ is bounded in probability. Consequently, $\sup_{a\in K} |T_{a1}|$ goes to zero in probability as~$n\to\infty$. Part~(i) of the result follows.

(ii)
Fix~$\delta>0$ and $\eta>0$. Writing $\hat q=\hatq$ and, as in the previous section, $\tilde q=\tildeq$, first choose $n_1$ and $M$ large enough to have $|\tilde q| \leq M$ and $P[ |\hat q|> M]<\eta/2$ for any~$n\geq n_1$ (Lemma~\ref{lemme_conv_p_i}(i) implies that $\hat q$ is the sample quantile of a number of~$Y_i$'s that increases to infinity, so that $|\hat q|$, with arbitrarily large probability for $n$ large, cannot exceed $2\sup_{x\in\suppX} |q_\alpha(x)|$).  Then, with $\mathbb I_+ = \mathbb I_{[\min_{a\in\R} \widehat G^{N,n}_{a}(\hat x^N) \ge \min_{a\in\R}  \widetilde G_{a}(\tilde{x}) ]}$, we have
\begin{align}
\lefteqn{
\hspace{-10mm}
|\min_{a\in\R} \widehat G^{N,n}_{a}(\hat x^N) - \min_{a\in\R}  \widetilde G_{a}(\tilde{x})| \mathbb I_+
=\big(\widehat G^{N,n}_{\hat q}(\hat x^N) - \widetilde G_{\tilde q}(\tilde{x})\big)\mathbb I_+
}
\nonumber
\\[3mm]
\hspace{10mm}
&
\le 
\big(\widehat G^{N,n}_{\tilde q}(\hat x^N) - \widetilde G_{\tilde q}(\tilde{x})\big)\mathbb I_+
\le \sup_{a\in[-M,M]}|\widehat G^{N,n}_a(\hat x^N) - \widetilde G_a(\tilde{x})|\mathbb I_+
,
\label{tdj1}
\end{align}
almost surely. Now, with $\mathbb I_- = \mathbb I_{[\min_{a\in\R} \widehat G^{N,n}_{a}(\hat x^N) < \min_{a\in\R}  \widetilde G_{a}(\tilde{x}) ]}$, we have that, under $|\hat q|\leq M$, 
\begin{align}
\lefteqn{
\hspace{-10mm}
|\min_{a\in\R} \widehat G^{N,n}_{a}(\hat x^N) - \min_{a\in\R}  \widetilde G_{a}(\tilde{x})| \mathbb I_-
= |\big(\widetilde G_{\tilde q}(\tilde{x}) - \widehat G^{N,n}_{\hat q}(\hat x^N)\big)\mathbb I_-
}
\nonumber
\\[3mm]
\hspace{10mm}
&
\le |\big(\widetilde G_{\hat q}(\tilde{x}) - \widehat G^{N,n}_{\hat q}(\hat x^N)\mathbb I_-
\le \sup_{a\in[-M,M]} |\widehat G^{N,n}_{a}(\hat x^N) - \widetilde G_a(\tilde{x})|\mathbb I_-
.
\label{tdj2}
\end{align}
By combining~(\ref{tdj1}) and~(\ref{tdj2}), we obtain that, under~$|\hat q|\leq M$, 
$$
|\min_{a\in\R} \widehat G^{N,n}_{a}(\hat x^N) - \min_{a\in\R}  \widetilde G_{a}(\tilde{x})|
\leq
\sup_{a\in[-M,M]} |\widehat G^{N,n}_{a}(\hat x^N) - \widetilde G_a(\tilde{x})|.
$$

Consequently, for any~$n\geq n_1$, we obtain
\begin{eqnarray*}
\lefteqn{
P\Big[
|\min_{a\in\R} \widehat G^{N,n}_{a}(\hat x^N) - \min_{a\in\R}  \widetilde G_{a}(\tilde{x})|
>\delta
\Big]
}
\\[2mm]
& & \hspace{2mm}
=
P\Big[
|\min_{a\in\R} \widehat G^{N,n}_{a}(\hat x^N) - \min_{a\in\R}  \widetilde G_{a}(\tilde{x})|
>\delta,|\hat q|\leq M
\Big]
+
P\Big[
|\min_{a\in\R} \widehat G^{N,n}_{a}(\hat x^N) - \min_{a\in\R}  \widetilde G_{a}(\tilde{x})|
>\delta,|\hat q|> M
\Big]
\\[2mm]
& & \hspace{2mm}
\leq
P\Big[
\sup_{a\in[-M,M]} |\widehat G^{N,n}_{a}(\hat x^N) - \widetilde G_a(\tilde{x})|>\delta
\Big]
+
\frac{\eta}{2}.
\end{eqnarray*}
From Part~(i) of the lemma, the first term is smaller than~$\eta/2$ for any~$n\geq n_2$.  We conclude that, for any~$n\geq n_0:=\max(n_1,n_2)$, we have 
$$
P\Big[|\min_{a\in\R} \widehat G^{N,n}_{a}(\hat x^N) - \min_{a\in\R}  \widetilde G_{a}(\tilde{x})|>\delta\Big]<\eta,
$$
which shows Part~(ii) of the result.

(iii) The proof proceeds in the same way as in (ii) above. First we pick $n_1$ and $M$ large enough to have $P[ |\hat q|> M]<\eta/2$ for any~$n\geq n_1$, which yields
\begin{equation}
P\Big[
|\widetilde G_{\tilde q}(\tilde{x}) - \widetilde G_{\hat q}(\tilde{x}) |
>\delta
\Big]
\leq
P\Big[
|\widetilde G_{\tilde q}(\tilde{x}) - \widetilde G_{\hat q}(\tilde{x}) |
>\delta
,
|\hat q|\leq M
\Big]
+
\frac{\eta}{2}
.
\label{allezlesbelges}
\end{equation}
Now, from the triangular inequality, we obtain
\begin{eqnarray*}
\lefteqn{
P\Big[
|\widetilde G_{\tilde q}(\tilde{x}) - \widetilde G_{\hat q}(\tilde{x}) |
>\delta,|\hat q|\leq M
\Big]
}
\\[2mm]
& & \hspace{5mm}
\leq 
P\Big[
|\widetilde G_{\tilde q}(\tilde{x}) - \widehat G^{N,n}_{\hat q}(\hat x^N)|
>\delta/2,|\hat q|\leq M
\Big]
+
P\Big[
|\widehat G^{N,n}_{\hat q}(\hat x^N)- \widetilde G_{\hat q}(\tilde{x})|
>\delta/2,|\hat q|\leq M
\Big]
\\[2mm]
& & \hspace{5mm}
\leq 
P\Big[
|\min_{a\in\R} \widehat G^{N,n}_{a}(\hat x^N) - \min_{a\in\R}  \widetilde G_{a}(\tilde{x})|
>\delta/2
\Big]
+
P\Big[
\sup_{a\in[-M,M]}  |\widehat G^{N,n}_{a}(\hat x^N) - \widetilde G_a(\tilde{x})|
>\delta/2
\Big],
\end{eqnarray*}
which, from Part~(i) and Part~(ii) of the lemma, can be made arbitrarily small for~$n$ large enough. Jointly with~(\ref{allezlesbelges}), this establishes the result. 
\end{proof}

We can now conclude with the proof of Theorem~\ref{consistth}. 

\begin{proof}[Proof of Theorem~\ref{consistth}]
Since the function $\rho_\alpha(\cdot)$ is strictly convex, $\widetilde G_a(\tilde{x})$ is also strictly convex in~$a$. Its minimum in~$a$ (for any fixed~$x$) is therefore unique, and the convergence in probability of~$\widetilde G_{\hat q}(\tilde{x})$ towards $\widetilde G_{\tilde q}(\tilde{x})$ implies the convergence in probability of the corresponding arguments.

\end{proof}


\bibliography{biblio_cond_qu2}

\end{document}